\documentclass[journal]{IEEEtran}

\IEEEoverridecommandlockouts
\makeatletter
\let\MYcaption\@makecaption
\makeatother

\usepackage{subcaption}

\usepackage{siunitx}
\usepackage{cite,comment}
\usepackage{amsmath,amssymb,amsfonts}
\usepackage{bm}
\usepackage{graphicx,enumitem}
\usepackage{textcomp}
\usepackage{xcolor,soul}
\usepackage{balance,mathcomSTEv4_v1}
\usepackage{booktabs}
\usepackage{cleveref}
\usepackage{pgfplots}
\usepackage{pgfplotstable} 

 \usepackage{algorithm}
 \usepackage[scaled=0.92]{helvet}
 \usepackage{algpseudocode}
 \usepackage{balance}
 \usepackage{tikz}
 \usepackage{xspace}
 \usepackage[acronym,shortcuts]{glossaries}
\crefname{section}{Sec.}{Secs.}
\crefname{figure}{Fig.}{Figs.}
\crefname{table}{Tab.}{Tabs.}
\crefname{equation}{Eq.}{Eqs.}

\newtheorem{lemma}{Lemma}
\usepgfplotslibrary{colormaps}
\usepackage{svg}

 \definecolor{bleudefrance}{rgb}{0.19, 0.55, 0.91}
 \definecolor{myDarkGreen}{RGB}{0,100,0} 
 \definecolor{myDarkGray}{gray}{0.25} 

\newtheorem{rem}{Remark}

 \definecolor{myDarkGreen}{RGB}{0,100,0} 
 \definecolor{myDarkBlue}{RGB}{0,0,139} 
 \definecolor{myDarkGray}{gray}{0.25} 

 \newcommand{\algo}{\textsf{FedSwitch}\xspace }

 \usetikzlibrary{positioning}

\makeatletter
\newcommand{\new}[1]{{\textcolor{blue}{#1}}}

\makeatletter
\newcommand\remembertext[2]{
  \immediate\write\@auxout{\unexpanded{\global\long\@namedef{mytext@#1}{#2}}}%
  {\color{black} #2}%
}
\newcommand\recalltext[1]{%
  \new{\ifcsname mytext@#1\endcsname
    \fontsize{10.5}{12.5}\selectfont\@nameuse{mytext@#1}%
  \else
    ``??''
  \fi
}}
\makeatother

\begin{document}
\sloppy
\newacronym{awgn}{AWGN}{additive white Gaussian noise}
\newacronym{5g}{5G}{fifth-generation}
\newacronym{cdf}{CDF}{cumulative distribution function}

\newacronym{simo}{SIMO}{single-input and multiple-output}
\newacronym{ml}{ML}{maximum likelihood}

\newacronym{pdf}{pdf}{probability density function}
\newacronym{psk}{PSK}{phase-shift keying}
\newacronym{nlos}{NLOS}{not line-of-sight}
\newacronym{skc}{SKC}{secret-key capacity}
\newacronym{pmd}{PMD}{probability mass distribution}
\newacronym{los}{LOS}{line-of-sight}
\newacronym{skg}{SKG}{secret-key generation}
\newacronym{ris}{RIS}{reflective-intelligent-surface}
\newacronym{cr}{CR}{challenge-response}
\newacronym{md}{MD}{miss-detection}
\newacronym{fa}{FA}{false-alarm}
\newacronym{pla}{PLA}{physical-layer authentication}
\newacronym{crpla}{CR-PLA}{challenge-response physical-layer authentication}
\newacronym{snr}{SNR}{signal-to-noise ratio}
\newacronym{std}{STD}{standard deviation}
\newacronym{aoa}{AOA}{angle-of-arrival}
\newacronym{ofdm}{OFDM}{orthogonal-frequency-division-multiplexing}
\newacronym{cnn}{CNN}{convolutional-neural-network}
\newacronym{llr}{LLR}{log-likelihood-ratio}
\newacronym{lr}{LR}{learning-rate}
\newacronym{cfr}{CFR}{channel-frequency-response}
\newacronym{ue}{UE}{user-equipment}
\newacronym{bs}{BS}{base-station}
\newacronym{fl}{FL}{federated learning}
\newacronym{fedavg}{FedAvg}{federated-averaging}
\newacronym{iot}{IoT}{Internet-of-Things}
\newacronym{ula}{ULA}{uniform-linear-array}
\newacronym{csi}{CSI}{channel state information}
\newacronym{nn}{NN}{neural-network}
\newacronym{ce}{CE}{cross-entropy}
\newacronym{jrc}{JRC}{joint-region-classification}
\newacronym{roo}{RoO}{Region-of-Origin}
\newacronym{ps}{PS}{parameter server}
\newacronym{pls}{PLS}{physical-layer-security}
\newacronym{lrt}{LRT}{likelihood-ratio-test}
\newacronym{llrt}{LLRT}{log-likelihood-ratio-test}
\newacronym{msgd}{mSGD}{momentum stochastic-gradient-descent}
\newacronym{det}{DET}{detection-error-tradeoff}
\newacronym{sgd}{SGD}{stochastic-gradient-descent}
\newacronym{np}{NP}{Neyman-Pearson}
\newacronym{mlr}{ML}{machine-learning}
\newacronym{rssi}{RSSI}{received-signal-strength-indicator}
\newacronym{crc}{CRC}{Channel-based Region Classification} 

\title{\algo: Federated Region Classification \\ From 
Wireless Channel Measurements 

\thanks{
M.Piana and S. Tomasin are with the Department of Information Engineering, University of Padova (UNIPD), Italy. S. Tomasin is also with the Department of Mathematics, University of Padova, Italy. This work was partially funded by the European Commission through the Horizon Europe/JU SNS project ROBUST-6G (Grant Agreement no. 101139068). S. Rini is with the Department of Electrical and Computer Engineering, National Yang-Ming Chiao-Tung University (NYCU), Hsinchu, Taiwan (e-mail: stefano.rini@nycu.edu.tw) and with the German Aerospace Center (DLR), Oberpfaffenhofen, Germany (e-mail stefano.rini@dlr.de). It was also supported in part by the NSTC grant 114-2923-E-A49 -006 -MY3 and 113-2923-E-A49 -001 -MY2 and by MARC, the MediaTEK Advanced Research Center, with grant number 114A540531. Corresponding author: Mattia Piana, email: mattia.piana@phd.unipd.it.
}}

\author{ Mattia~Piana,~\IEEEmembership{Graduate Student~Member,~IEEE}, \IEEEauthorblockN{Stefano~Rini,~\IEEEmembership{Senior~Member,~IEEE}, and Stefano~Tomasin,~\IEEEmembership{Senior~Member,~IEEE}}}

\maketitle

\begin{abstract}
Wireless \ac{iot} networks can leverage locally observed \ac{csi} at \acp{bs} to perform \emph{region classification}, i.e., infer the region of origin of a transmitter (TX), thereby enabling channel-based authentication.
However, centralizing these measurements incurs substantial communication overhead and requires sharing location-dependent radio fingerprints.
To address these difficulties, we propose an edge-native \ac{fl} framework where \acp{bs} collaboratively train a shared classifier via a \ac{ps} without exposing raw data, thereby supporting physical-layer security tasks like spoofing detection. 
To overcome the statistical heterogeneity arising from diverse \ac{bs} locations and propagation conditions, we introduce \algo, a personalized \ac{fl} algorithm that automatically transitions edge nodes from collaborative training to local fine-tuning upon loss stagnation. 
We provide theoretical convergence guarantees for the proposed collaborative training phase and derive analytical bounds linking classification errors to statistical discrepancies. 
Extensive evaluations on synthetic channel measurements demonstrate that \algo significantly improves region-classification accuracy over standard \ac{fl} baselines. 
\end{abstract}

\glsresetall



\section{Introduction}
\label{sec:intro}
%
Wireless \ac{iot} networks often comprise resource-constrained devices, motivating lightweight security mechanisms that complement conventional cryptographic authentication. 
Since wireless channel characteristics depend on the transmitter location, locally observed \ac{csi} can provide physical-layer evidence of whether a transmitter is spatially consistent with an expected region.

In the following, we consider a scenario in which multiple \acp{bs} collect
channel-dependent measurements from users.
We formulate the \ac{crc} problem, where each \ac{bs} uses local channel
observations to classify a TX into predefined regions, supporting security
tasks such as location-aware authentication and rogue-transmitter detection.
%
\remembertext{rev2.6}{
We adopt an \ac{fl} framework in which
\acp{bs} exchange model updates while keeping raw channel measurements local.
Since different propagation conditions induce heterogeneous local datasets,
a single global model may be suboptimal, motivating a personalized federated
approach.}
We propose \algo, which combines collaborative training across \acp{bs} with
local fine-tuning. Evaluations on 3GPP-compliant synthetic channel measurements
show that \algo outperforms standard federated baselines in heterogeneous
deployments.

\smallskip
\noindent
\textbf{{Relevant Literature:}}
We briefly review the main research directions that are relevant to the development of this work.

\noindent
\emph{Region Verification} A simple case of region localization is the problem of determining if a device is in a \ac{roo} or not, which is a binary classification problem.
Such a scenario has been considered in \cite{8798661}, where the authors showed that \ac{nn} trained with a sufficiently high number of samples achieve the performance of the optimal \ac{np}-test. 
In \cite{10005038}, the authors propose FedPos, in which devices localize the transmitter in an indoor environment using its fingerprint.
In particular, they tackle data heterogeneity by training in an \ac{fl} manner only the feature extraction part of the model. On the same line, \cite{10628047} exploits meta-learning to first learn a shared representation across clients, followed by an adaptation phase. Finally, \cite{tasbaz2024feature} proposes to use both \ac{csi} and \ac{rssi} along with \ac{fedavg} to make the system more robust to fluctuations.

\noindent
\emph{Multiple-Stage Federated Learning} The idea of learning a task in multiple stages, usually with a first \ac{fl} stage to learn a shared representation used as a warm starting point, followed by a second stage specifically tailored for the problem at hand, is well established in the literature 
\cite{cheng2021fine,chen2024robust,peng2024fedpft}. In \cite{xu2022fedcorr}, the authors tackle the label noise problem with FedCorr, which first identifies noisy clients, then fine-tunes the global model on only the clean clients, and finally corrects the remaining noisy labels after fine-tuning. Dealing with the label flipping attack, the authors in \cite{11129099} propose FedMAR, a multistage asynchronous aggregation framework that identifies participants with poisoned data and mitigates their negative impact through a roll-back approach. In \cite{wang2023fedbfpt}, the authors show that pretraining using federated methods is effective, especially in healthcare systems that suffer from data privacy issues \cite{liu2020federated}. In \cite{ZHU2026115425}, we have FedERFT, a federated learning framework with server-side staged optimization,  a two-stage approach with, in the first stage, the server aggregates client models to obtain an initial global model. In the second stage, the global encoder is frozen, and only the classification head is fine-tuned, and in the final stage, both the global encoder and the classification head are jointly fine-tuned to further enhance overall performance. In \cite{chen2024feddat}, we have FedDAT, which leverages a Dual-Adapter Teacher to address data heterogeneity by regularizing the client local updates and applying mutual knowledge distillation for an efficient knowledge transfer.

\paragraph*{Physical-Layer Authentication} A problem connecting to region classification is \ac{pla}, where messages are authenticated by looking at the channel characteristics over which the signal is travelling. Assuming to know the position or the region where the legitimate transmitter should be, the receiver can authenticate messages by checking if the channel corresponds to the expected position or region. 
The application of \ac{fl} in \ac{pla} is at its dawn, and to the best of our knowledge, only a few works have been done. In \cite{wang2021collaborative}, we have an \ac{iot} scenario where multiple constrained devices collaborate to authenticate the transmitter using the \ac{fedavg} algorithm. In \cite{halder2023radio}, a distributed anomaly detection system for detecting compromised devices in LoRa-enabled industrial IoT is presented. It exploits device-specific features such as Carrier Frequency Offset (CFO) as the device fingerprint, whose deviations from expected behaviors allow for the detection of attackers.
In \cite{nagia2022federated}, multiple Wi-Fi routers use FedProx \cite{li2020federated} to cooperate in localizing the TX. Finally, in \cite{yazdinejad2021federated}, the authors tackle drone authentication using \ac{fedavg} at the local server and \ac{sgd} performed locally by the drones.


\smallskip
\noindent
\textbf{{Contributions:}}
The main contributions of this paper are:
\begin{itemize}[leftmargin=*]
    \item {Sec. \ref{sec:syst_model}} -- \textbf{
    \ac{crc}
    model for wireless \ac{iot}.}
    We formulate a \ac{crc} problem for wireless \ac{iot} networks, in which multiple \acp{bs} observe uplink channel measurements generated by users transmitting from different physical regions.
    Each \ac{bs} aims to infer the \ac{roo} of a transmission from its local channel observations, so that the resulting estimate can be used as a lightweight physical-layer security mechanism.
    %

    \item  {Sec. \ref{sec:Problem Definition}} -- \textbf{Heterogeneus \ac{fl} formulation.}
    We cast the \ac{crc} task as a \ac{fl} problem, where the \acp{bs} collaboratively train a classifier under the coordination of a \ac{ps} without sharing their recorded channel measurements.
    This formulation avoids the centralization of location-dependent radio fingerprints, reduces the exposure of sensitive channel data, and enables collaborative learning across edge nodes while preserving the distributed nature of the wireless \ac{iot} deployment.

    \item {Sec. \ref{sec:fedloss}} -- \textbf{Proposed Solution: \algo.}
    We introduce \algo, a personalized \ac{fl} algorithm designed for the heterogeneous channel statistics induced by multi-\ac{bs} deployments.
    The proposed method allows each \ac{bs} to participate in collaborative training while the global model is beneficial, and to switch to a local fine-tuning phase when its training loss stagnates.
    In this way, \algo balances the benefits of shared learning with the need for local specialization at edge nodes observing different propagation conditions and region distributions.

    \item {Sec. \ref{sec:Performance Analysis}} -- \textbf{Theoretical analysis of learning and classification performance.}
    We provide a theoretical characterization of the proposed learning framework, including convergence properties of \algo and analytical bounds relating the \ac{crc} performance to the statistical discrepancy among local channel distributions.
    This analysis clarifies the role of data heterogeneity in federated channel-based classification and motivates the use of personalization in the considered wireless \ac{iot} setting.

    \item {Sec. \ref{sec:setup}} -- \textbf{Evaluation on synthetic vehicular channel measurements.}
    We validate the proposed framework using synthetic channel measurements from the 3GPP channel model.
    The numerical results show that channel measurements can support accurate \ac{crc} and that \algo improves over standard \ac{fl} baselines in heterogeneous multi-\ac{bs} deployments.
\end{itemize}

A conceptual representation of the contribution of the paper is provided in Fig. \ref{fig:high level}.

\smallskip
\noindent
{\emph{Notation:}}
Calligraphic letters denote sets. 

\begin{figure*}
    \centering
\includegraphics[width=.8\linewidth]{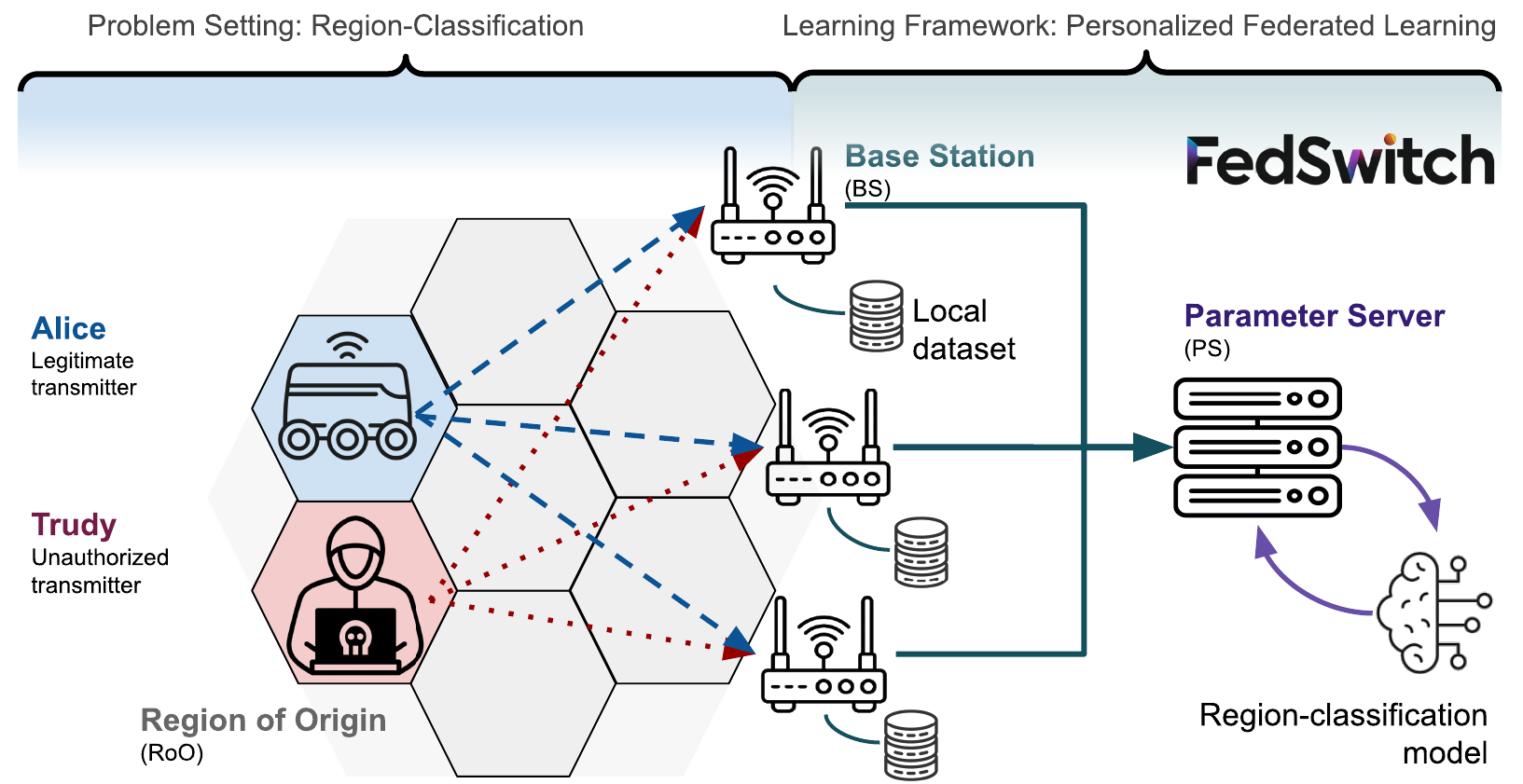}
    %
    \caption{
Graphical overview of the proposed \algo framework.
In the considered wireless \ac{iot} setting, users transmit from predefined \acp{roo} and multiple \acp{bs} observe location-dependent channel measurements.
The \acp{bs} use these measurements to infer whether a received signal is spatially consistent with the claimed or authorized region, enabling lightweight physical-layer verification against protocol-compliant impersonation attempts.
Rather than centralizing raw channel measurements, the \acp{bs} collaboratively train a region-classification model through a \ac{ps}.
The proposed \algo algorithm first exploits federated training to learn shared channel representations and then switches to personalized local fine-tuning when further global updates no longer improve the local loss.
}
\label{fig:high level}
\end{figure*}

\section{Problem Setting:  \\ 
Channel-based Region Classification}
\label{sec:syst_model}

We consider a wireless \ac{iot} scenario in which devices transmit uplink signals to a set of $K$ geographically distributed \acp{bs}, indexed by
$k \in \{1,\ldots,K\} \triangleq [K]$.
Each \ac{bs} observes channel-dependent measurements and uses them to infer the \ac{roo} of the received transmission.
The resulting \ac{crc} task is motivated by physical-layer security applications, where a transmitter is expected to operate from an authorized region and a mismatch between the claimed and inferred \ac{roo} may indicate spoofing or impersonation.
We denote by Alice a legitimate transmitter associated with an authorized \ac{roo}, and by Trudy a protocol-compliant unauthorized transmitter attempting to impersonate Alice from a different region.
a pa

Let $\mathcal{A}\subset\mathbb{R}^2$ denote the geographical area of interest.
We partition $\mathcal{A}$ into $N$ predefined \acp{roo},
\begin{equation}
    \mathcal{A}
    =
    \bigcup_{n\in [N]} \mathcal{R}_n,
    \qquad
    \mathcal{R}_n \cap \mathcal{R}_{n'} = \emptyset,
    \quad n\neq n',
    \label{eq:region_partition}
\end{equation}
where each region $\mathcal{R}_n$ corresponds to one possible \ac{roo} label.
Each \ac{bs} observes uplink transmissions generated by users located in one of the regions $\{\mathcal{R}_n\}$.
We assume that the transmitters are equipped with a single antenna, while the \acp{bs} are equipped with $M$ antennas, consistent with the measurement setting considered in this work. 
When a user transmits from position $\bm{p}\in\mathcal{R}_n$, \ac{bs} $k$ observes a noisy channel response that depends on the transmitter position, the \ac{bs} position, and the surrounding propagation environment.
For subcarrier $s\in [S]$, we denote the corresponding complex channel coefficient by $\bm{h}_{k,s}(\bm{p}) \in \mathbb{C}^M$.
The multi-subcarrier channel observation collected at \ac{bs} $k$ is the column vector
\begin{equation}
    \bm{x}_k(\bm{p})
    =
    \big[
        \bm{h}^T_{k,1}(\bm{p}),
        \bm{h}^T_{k,2}(\bm{p}),
        \ldots,
        \bm{h}^T_{k,S}(\bm{p})
    \big]^\mathsf{T}
    \in \mathbb{C}^{MS}.
    \label{eq:channel_observation}
\end{equation}
The detailed channel-generation model and the construction of the channel observations used in the experiments are described later in Sec.~\ref{sec:setup}.
In the present section, we only require that $\bm{x}_k(\bm{p})$ contains location-dependent information induced by the propagation environment. In this paper, we are particularly interested in the scarce-data regime, a typical scenario in IoT networks where constrained devices cannot perform channel estimation frequently.

The \ac{crc} problem at \ac{bs} $k$ consists of assigning the received channel observation $\bm{x}_k(\bm{p})$ to one of the $N$ predefined regions.
That is, \ac{bs} $k$ implements a decision rule
\begin{equation}
    \delta_k : \mathbb{C}^{MS} \rightarrow [N],
    \label{eq:region_classifier_decision}
\end{equation}
where $\delta_k(\bm{x}_k)=n$ means that the signal observed at \ac{bs} $k$ is classified as originating from region $\mathcal{R}_n$.
\subsection{Performance Evaluation} \label{sub: perf_eval}

%
Let $P_{\rm MD}(\delta_k | \Hcal_n)$/$P_{\rm FA}(\delta_k | \Hcal_n)$ as the probability of \ac{md}/\ac{fa} under the hypothesis $\Hcal_n$, which corresponds to the region $\mathcal{R}_n$. The performance of the decision rule in \eqref{eq:region_classifier_decision} 
 is evaluated through the \ac{det} defined as
\ea{
\mathrm{DET}(\delta_k,\Hcal_n )=\lb P_{\rm MD}(\delta_k | \Hcal_n),P_{\rm FA}(\delta_k | \Hcal_n)  \rb \,, 
\label{eq:ROC}
}
which is a function both of the decision rule $\delta_k$ and of the specific \ac{roo}.
We define the average over both \acp{bs} and true \ac{roo} \ac{md}/\ac{fa} probabilities as
\eas{
\bar{P}_{\rm MD} & =\frac{1}{KN} \sum_{n,k \in [N],[K]}P_{\rm MD}(\delta_k | \Hcal_n) \\
\bar{P}_{\rm FA} & =\frac{1}{KN} \sum_{n,k \in [N],[K]} P_{\rm FA}(\delta_k | \Hcal_n)\,.
}{\label{eq:avg md-fa}}
The average \ac{det} across all hypotheses and \acp{bs} is then the tuple
\ea{
\overline{\mathrm{DET}}=\lb \bar{P}_{\rm MD},   \bar{P}_{\rm FA}\rb  
\label{eq:avg DET}
}
The DET curve in \eqref{eq:avg DET} can be used as a lightweight physical-layer verification mechanism.
For example, if Alice is authorized to transmit from region $\mathcal{R}_n$, then a classification outcome different from $n$ provides evidence that the received signal may be spatially inconsistent with Alice's authorized region.
Equivalently, if a received transmission claims to originate from $\mathcal{R}_n$, the \ac{bs} can use the output of $\delta_k$ to accept or reject this claim based on the inferred \ac{roo}. With \eqref{eq:avg DET} we are capturing the average across \acp{bs} and \acp{roo} of such quantity.

\section{Learning Framework: \\ Personalized Federated Learning }
\label{sec:Problem Definition}

Based on the \ac{crc} problem introduced in Sec.~\ref{sec:syst_model}, we now formulate the learning problem solved by the \acp{bs}.
The objective is to train \ac{nn} classifiers that infer the \ac{roo} of a received signal while keeping the raw channel measurements local at each \ac{bs}.
This naturally leads to a \ac{fl} formulation, in which the \acp{bs} collaborate through a \ac{ps} without sharing their recorded channel data.

\subsection{Local Datasets and Classification Loss}

Each \ac{bs} $k\in[K]$, collects a local labeled dataset
\begin{equation}
    \mathcal{D}_k
    =
    \left\{
    \left( \bm{x}_k^{(i)}, y_k^{(i)}\right)
    \right\}_{i \in [C_k]},
    \label{eq:local_dataset}
\end{equation}
where $ \bm{x}_k^{(i)}$ is a realization of $\bm{x}_k(\bm{p})$, the multi-subcarrier channel observation collected at \ac{bs} $k$, as defined in \eqref{eq:channel_observation}, and $y_k^{(i)} \in [N]$ is the corresponding \ac{roo}.
The number of samples available at \ac{bs} $k$ is denoted by $C_k$, and the total number of samples across all \acp{bs} is
$C = \sum_{k \in [K]} C_k$.
Each \ac{bs} trains a classifier
$
    f_k(\cdot;\bm{w}) : \mathbb{C}^{MS} \rightarrow \Delta^{N}
$, where $\bm{w}\in\mathbb{R}^{m}$ denotes the model parameters and $\Delta^{N}$ is the probability simplex over the $N$ region labels.
For an input channel observation $\bm{x}_k^{(i)}$, the classifier output
\begin{equation}
    f_k(\bm{x}_k^{(i)};\bm{w})
    =
    \left[
    q_1(\bm{x}_k^{(i)};\bm{w}),
    \ldots,
    q_N(\bm{x}_k^{(i)};\bm{w})
    \right]
\end{equation}
represents the estimated posterior probabilities of the $N$ \acp{roo}.  
The predicted region is obtained as
\begin{equation}
    \widehat{y}_k
    =
    \arg\max_{n\in[N]}
    q_n(\bm{x}_k^{(i)};\bm{w}).
\end{equation}

The local empirical risk at \ac{bs} $k$ is defined as
\begin{equation}
    F_k(\bm{w})
    =
    \frac{1}{C_k}
    \sum_{i \in [C_k]}
    \ell
    \left(
        f(\bm{x}_k^{(i)};\bm{w}),
        y_k^{(i)}
    \right),
    \label{eq:local_empirical_risk}
\end{equation}
where $\ell(\cdot,\cdot)$ is the multi-class cross-entropy loss.
If all channel measurements could be centralized, the network would train a model by minimizing the global empirical risk
\begin{equation}
    F(\bm{w})
    =
    \sum_{k \in [K]}
    \frac{C_k}{C}
    F_k(\bm{w}).
    \label{eq:global_empirical_risk}
\end{equation}
However, centralizing the datasets $\{\mathcal{D}_k\}_{k \in [K]}$ is undesirable because the channel observations are location-dependent radio fingerprints.
We therefore minimize \eqref{eq:global_empirical_risk} in a federated manner.

\subsection{Federated Learning Preliminaries}
\label{sec:preliminary}

Consider the FL setting with $K$ clients, each client $k$ possessing a local dataset
${\mathcal{D}}_k$, wishing to minimize the \emph{loss function} $\mathcal{L}$ as evaluated across all the clients and over the model weights $\bm{w}$. 
This minimization is coordinated by the PS as follows: at communication round $e \in [T_{\rm E}]$, where $T_{\rm E}$ is the total number of communication rounds, the \ac{ps} broadcasts the current global model $\bm{w}^{(e)}$ to a subset of participating \acp{bs}.
Each selected \ac{bs} initializes its local model with $\bm{w}^{(e)}$ and performs local training on $\mathcal{D}_k$.
Denoting by $\bm{w}_k^{(e+1)}$ the locally updated model at \ac{bs} $k$, the \ac{ps} aggregates the received models as
\begin{equation}
    \bm{w}^{(e+1)}
    =
    \sum_{k\in\mathcal{S}^{(e)}}
    \frac{C_k}{\sum_{j\in\mathcal{S}^{(e)}} C_j}
    \bm{w}_k^{(e+1)},
    \label{eq:fedavg_update}
\end{equation}
where $\mathcal{S}^{(e)}\subseteq[K]$ is the set of \acp{bs} participating in round $e$.
This corresponds to the standard \ac{fedavg} update.

While \ac{fedavg} is effective when local datasets are approximately identically distributed, this assumption is not satisfied in the considered region-classification problem.
Indeed, each \ac{bs} observes the same geographical area, but from different positions and under different propagation conditions. As a result, the local risks $F_k(\bm{w})$ in \eqref{eq:local_empirical_risk} may have different minimizers, and the single global model provided in \eqref{eq:global_empirical_risk} may not provide the best classifier for every \ac{bs}.

In practical FL scenarios, the communication between clients and the PS is often limited.
%
One common constraint considered in the literature on communication between clients and the PS is the \emph{rate constraint}-- in which the communication is restricted to $m \mathrm{R}_{k}$ bits from each client to the PS.
%
To meet this constraint, 
each client $k$ compresses the $m$-dimensional gradient vector $\bm{g}_{k}^{(e)}$ to a $m \mathsf{R}_k$ binary vector through a compressor $\mathrm{comp}_{\mathrm{R}_k}: \mathbb{R}^m \rightarrow [2^{m \mathsf{R}_k}]$,
where the sum over all $\mathrm{R}_k$ is constrained,
and the client is dropped (unselected) if assigned a rate of $\mathrm{R}_k=0$.
The PS aggregates all the compressed gradients and forms the new global weights
\begin{equation}
    \bm{w}^{(e+1)} = \bm{w}^{(e)} - \eta_{e} \hat{\bm{g}}^{(e)},
    \text{ for $e\in [T_{\rm E}]$, where}
\label{eq:glob_update}
\end{equation}
\begin{equation}
    \hat{\bm{g}}^{(e)} = \frac{1}{K} \sum_{k \in [K]} \mathrm{comp}^{-1}_{\mathrm{R}_k} (\mathrm{comp}_{\mathrm{R}_k}(\bm{g}_{k}^{(e)})).
\label{eq:aggregation}
\end{equation}
The ultimate goal in a single FL iteration is to minimize the distortion in gradient reconstruction, i.e., the distortion between $\bm{g}^{(e)}$ and $\hat{\bm{g}}^{(e)}$.
We assume that minimizing the gradient reconstruction distortion leads to a minimum perturbation of the accuracy.



%
\subsection{Motivation for Personalized Federated Learning}

The heterogeneity of the local channel distributions makes indefinite collaborative training suboptimal.
At early training stages, model updates from other \acp{bs} are beneficial, especially when local datasets are small.
As training progresses, however, the global model may become less representative of the propagation conditions observed at a given \ac{bs}, degrading its local classification performance.
This motivates a personalized \ac{fl} strategy that combines initial collaboration with subsequent local adaptation.

Each \ac{bs} should initially benefit from shared training across the network, but should also be allowed to specialize its classifier once the global updates stop improving its local objective.
The goal is therefore to learn a collection of personalized models $
    \{\bm{w}^\star_k\}_{k\in [K]}=\{\bm{w}_1^\star,\ldots,\bm{w}_K^\star\}
$,
where $\bm{w}_k^\star$ is adapted to the channel statistics and region distribution observed at \ac{bs} $k$.
This leads to the personalized objective of minimizing the global risk \eqref{eq:global_empirical_risk}
\begin{equation}
    \min_{\bm{w}_1,\ldots,\bm{w}_K}
    \sum_{k \in [K]}
    \frac{C_k}{C}
    F_k(\bm{w}_k),
    \label{eq:personalized_objective}
\end{equation}
with the additional requirement that the models should still exploit the information shared during the federated training phase.

\begin{table}[t]
\centering
\caption{Main Notation for Personalized \ac{fl}.}
\label{tab:fl_notation}
\begin{tabular}{lp{6.1cm}}
\toprule
\textbf{Symbol / Term} & \textbf{Description} \\
\midrule

$b_k$ & \textbf{Base Station:} $k$-th \ac{bs} \\

$\mathcal{R}_n$ & \textbf{Region of Origin:} $n$-th \ac{roo} \\


$\bm{x}_k^{(i)}$ & \textbf{Sample:} $i$-th channel observation at $b_k$ \\

$y_k^{(i)}$ & \textbf{Label:} \ac{roo} label associated with $\bm{x}_k^{(i)}$ \\

$\delta_k(\cdot)$ & \textbf{Local Decision Rule:} classifier used at $b_k$ \\

$\mathcal{D}_k$ & \textbf{Local Dataset:} Local dataset at $b_k$ \\

$e$ & \textbf{Communication round:} $e$-th comm. round out of $T_{\rm E}$ \\

$C_k$ & \textbf{Local Dataset Size:} Number of samples in $\mathcal{D}_k$ \\

$f(\cdot;\bm{w})$ & \textbf{Region Classifier:} Learned classifier with param. $\bm{w}$ \\

$w_k$ & \textbf{Local Model:} Classifier parameters trained at $b_k$ \\

$\bm{w}=\{w_k\}$ & \textbf{Distributed Model:} Collection of local models \acp{bs} \\

$F_k(\bm{w})$ & \textbf{Local Objective:} Empirical loss at $b_k$ \\

$F(\bm{w})$ & \textbf{Global Objective:} Global federated objective \\

$w_{\mathrm{F}}^{(e)}$ & \textbf{Global FL Model:} Global model at round $e$ \\

$\bm{w}_k^{(e,\ell)}$ & \textbf{Local Model:} Model at $b_k$ at round $e$ at local step $\ell$ \\


\ac{ps} & \textbf{Parameter Server:} Entity combining \acp{bs}' models \\

Alice & \textbf{Legitimate TX:} Device transmitting from its \ac{roo}. \\

Trudy & \textbf{Adversary:} Attacker   \\


\bottomrule
\end{tabular}
\end{table}





\section{Proposed Solution: \algo}
\label{sec:fedloss}

\algo combines collaborative federated training with local specialization.
In the first phase, the \acp{bs} collaboratively train a shared model using
\ac{fedavg}, exploiting channel measurements collected across the deployment.
In the second phase, each \ac{bs} progressively adapts the shared model to its
local propagation conditions and region coverage.

The transition between the two phases is determined adaptively from the local
training loss.
Each \ac{bs} evaluates the current global model on its local dataset and is
marked as ready for personalization when the loss improvement over a prescribed
observation window falls below a tolerance.
Once all \acp{bs} are marked as ready, the final global model initializes the
personalization phase.

Let $\bm{w}_{\rm F}^{(e)}$ denote the global model at communication round $e$, and let $\bm{w}_k^{(e)}$ denote the local model at \ac{bs} $k$.
We denote by $F_k(\bm{w})$ the empirical loss of model $\bm{w}$ on the local dataset $\mathcal{D}_k$.
For each \ac{bs}, we introduce a binary state variable $
    \rho_k^{(e)} \in \{0,1\}\,,
$ where $\rho_k^{(e)}=0$ means that \ac{bs} $k$ is still in the collaborative phase, while $\rho_k^{(e)}=1$ means that it is ready for local fine-tuning. Note that even if a \ac{bs} is ready for finetuning, it keeps participating in the \ac{fedavg} phase until all the others are ready: even if this can be a waste of resources for such \ac{bs}, it is not for the others, thus the whole system benefits from it. 
Given a patience window $T_{\rm W}$ and a tolerance $\epsilon>0$, the switching condition for the $k$-th \ac{bs} is triggered at round $e_k$ when, for the first time, the windowed average loss improvement follows below the threshold $\epsilon$, that is
\begin{equation} \label{eq:loss_stagnation_rule}
{\Delta F}_k^{(e)}
= \frac{1}{T_{\rm W}}\sum_{e=e_k-T_{\rm W}}^{e_k-1}
F_k(\bm{w}_k^{(e-1)})
- F_k(\bm{w}_k^{(e)})
<
\epsilon.
\end{equation}
When \eqref{eq:loss_stagnation_rule} holds, the loss improvement obtained from the global model over the last $T_{\rm W}$ communication rounds is considered negligible, and \ac{bs} $k$ is marked as ready for personalization. \remembertext{rev2.2.1}{ 
Note that the switching condition in \eqref{eq:loss_stagnation_rule} can be negative: it means that not only is training useless but is also damaging the \ac{bs} $k$, thus all the more so must switch.
}

\subsection{Collaborative \ac{fedavg} Phase}
At each communication round $e$, the \ac{ps} broadcasts the current global model $\bm{w}_{\rm F}^{(e)}$ to the participating \acp{bs}.
Each \ac{bs} initializes its local model as
$    \bm{w}_k^{(e,0)} = \bm{w}_{\rm F}^{(e)}
$ and performs $E$ local \ac{msgd} steps on $\mathcal{D}_k$.
For local step $\ell=0,\ldots,E-1$, the update is
\begin{equation}
    \bm{w}_k^{(e,\ell+1)}
    =
    \bm{w}_k^{(e,\ell)}
    -
    \eta^{(e,\ell)} \bm{m}_k^{(e,\ell)},
    \label{eq:local_update}
\end{equation}
where $\eta^{(e,\ell)}$ is the \ac{lr} and $\bm{m}_k^{(e,\ell)}$ is the momentum-based update direction.
Specifically,
\begin{equation}
    \bm{m}_k^{(e,\ell)}
    =
    \beta \bm{m}_k^{(e,\ell-1)}
    +
    (1-\beta) \bm{g}_k^{(e,\ell)},
    \label{eq:momentum_update}
\end{equation}
where $\beta\in[0,1)$ is the momentum coefficient and
\begin{equation}
    \bm{g}_k^{(e,\ell)}
    =
    \frac{1}{B_{\rm S}}
    \sum_{i\in\mathcal{B}_k^{(e,\ell)}}
    \nabla_w
    \ell
    \left(
        f(\bm{x}_k^{(i)};\bm{w}_k^{(e,\ell)}),
        y_k^{(i)}
    \right)
    \label{eq:mini_batch_gradient}
\end{equation}
is the stochastic gradient computed on a mini-batch $\mathcal{B}_k^{(e,\ell)}\subset\mathcal{D}_k$ of size $B_{\rm S}$.

After the $E$ local steps, \ac{bs} $k$ obtains $
    \bm{w}_k^{(e+1)} = \bm{w}_k^{(e,E)}$.
The \ac{ps} aggregates the local models according to the \ac{fedavg} rule
\begin{equation}
    \bm{w}_{\rm F}^{(e+1)}
    =
    \sum_{k \in [K]}
    \frac{C_k}{C}
    \bm{w}_k^{(e+1)},
    \qquad
    C=\sum_{k \in [K]} C_k.
    \label{eq:fedavg_fedloss}
\end{equation}
The updated global model is then evaluated locally by each \ac{bs} through its own loss $F_k(\bm{w}_{\rm F}^{(e+1)})$, and the switching rule in \eqref{eq:loss_stagnation_rule} is applied.

\subsection{Personalized Fine-Tuning Phase}
Let $T_{\rm s}$ denote the first communication round at which all \acp{bs} are marked as ready for personalization, i.e.,
\begin{equation}
    T_{\rm s}
    =
    \min
    \left\{
        e :
        \rho_k^{(e)}=1,\ \forall k\in[K]
    \right\}.
\end{equation}
At this point, the collaborative phase transitions into the personalization phase. \remembertext{rev2.3}{ 
Waiting for all the \acp{bs} to complete the information sharing is influenced by the slowest \ac{bs}, and it is actually a design choice. In fact, unlike typical \ac{fl} applications, our scenario has few clients, and all of them must perform well because each \acp{bs} makes its decision autonomously. 
} Each \ac{bs} initializes its personalized model from the final global model obtained after the collaborative stage, $
\widetilde{\bm{w}}_k^{(0)}
=
\bm{w}_{\rm F}^{(T_{\rm s})}$. 

After the switch, the \acp{bs} continue updating their models using both the local and global information. In particular, at each iteration, the model update is obtained as a weighted combination of the locally adapted model and the global model, namely
\begin{equation}
\widetilde{\bm{w}}_k^{(t)}
=
\alpha_t \bm{w}_k^{(t)}
+
(1-\alpha_t) \bm{w}_{\rm F}^{(t)},
\end{equation}
where $\alpha_t \in [0,1]$ is an exponentially increasing parameter controlling the tradeoff between personalization and global cooperation: as time passes, the contribution of the global model progressively decreases, while the local model becomes more influential, thus the models quickly but smoothly switch to pure fine-tuning. Note that also during the finetuning phase, each \acp{bs} keeps forwarding its models so that the \ac{ps} can compute the average model $\bm{w}_{\rm F}^{(t)}$. In practice, each \ac{bs} performs stochastic-gradient updates on its local dataset $\mathcal{D}_k$, while still incorporating information from the global model through the weighted combination above. This strategy preserves the representations learned during the collaborative phase while gradually adapting the classifier to the channel statistics observed by \ac{bs} $k$. The resulting procedure is summarized in Algorithm~\ref{alg:FedLoss}. The algorithm returns a set of personalized classifiers $\{\bm{w}_k^\star\}_{k \in [K]}$, one for each \ac{bs}, that are then used locally to infer the \ac{roo} of newly received transmissions.

\begin{algorithm}
\caption{
The proposed solution: \algo. 
}
\label{alg:FedLoss}
\begin{algorithmic}[1]
\Require Local datasets $\{\mathcal{D}_k\}_{k\in [K]}$, initial global model $\bm{w}_{\rm F}^{(0)}$, local steps $E$, momentum coefficient $\beta$, patience window $T_{\rm W}$, tolerance $\epsilon$, total number of communication rounds $T_{\rm E}$
\Ensure Personalized models $\{\bm{w}_k^\star\}_{k\in [K]}$

\State Initialize state variables $\rho_k^{(0)} \gets 0$ for all $k \in \{1, \dots, K\}$
\State Set communication round $e \gets 0$

\Repeat \hfill \Comment{\textbf{Phase 1: Collaborative FedAvg Phase}}
    \For{each \ac{bs} $k$ \textbf{in parallel}}
        \State Initialize local model: $\bm{w}_k^{(e,0)} \gets \bm{w}_{\rm F}^{(e)}$
        \For{local step $\ell = 0$ to $E-1$}
            \State Compute stochastic gradient $\bm{g}_k^{(e,\ell)}$ via \eqref{eq:mini_batch_gradient}
            \State Update momentum direction $\bm{m}_k^{(e,\ell)}$ via \eqref{eq:momentum_update}
            \State Update local parameters $\bm{w}_k^{(e,\ell+1)}$ via \eqref{eq:local_update}
        \EndFor
        \State Set local round output: $\bm{w}_k^{(e+1)} \gets \bm{w}_k^{(e,E)}$
    \EndFor
    
    \State \ac{ps} updates global model $\bm{w}_{\rm F}^{(e+1)}$ via \eqref{eq:fedavg_fedloss}
    \State $e \gets e + 1$
    
    \For{each \ac{bs} $k$ in parallel}
        \If{$e \ge T_{\rm W}$ \textbf{and} switching condition in \eqref{eq:loss_stagnation_rule} occurs}
            \State Mark as ready for personalization: $\rho_k^{(e)} \gets 1$
        \EndIf
    \EndFor
\Until{$\rho_k^{(e)} = 1$ for all $k \in \{1, \dots, K\}$}

\State Set switching round $T_{\rm s} \gets e$
\State Compute aggregation weights schedule 
$\{\alpha_t\}_{t=1}^{T_{\rm E}-{T_{\rm s}}}$
\Statex

\For{each \ac{bs} $k$ \textbf{in parallel}}
    \State Initialize personalized tracker: $\widetilde{\bm{w}}_k^{(0)} \gets \bm{w}_{\rm F}^{(T_{\rm s})}$
\EndFor
\Statex

\For{iteration $t = 1$ to $T_{\rm E}-T_{\rm s}$} \hfill \Comment{\textbf{Phase 2: Personalized Fine-Tuning Phase}}
    \For{each \ac{bs} $k$ \textbf{in parallel}}
        \State Perform \ac{msgd} updates on $\mathcal{D}_k$ starting from $\widetilde{w}_k^{(t-1)}$ to obtain $\bm{w}_k^{(t)}$
    \EndFor
    
    \State \ac{ps} updates global model $\bm{w}_{\rm F}^{(t)}$ by aggregating local updates $\bm{w}_k^{(t)}$ \hfill 
    
    \For{each \ac{bs} $k$ \textbf{in parallel}}
        \State Apply weighted combination: $\widetilde{\bm{w}}_k^{(t)} \gets \alpha_t \bm{w}_k^{(t)} + (1-\alpha_t) \bm{w}_{\rm F}^{(t)}$
    \EndFor
\EndFor

\For{each \ac{bs} $k$ \textbf{in parallel}}
    \State Save final personalized classifier: $\bm{w}_k^\star \gets \widetilde{w}_k^{(T_{\rm E}-T_{\rm s})}$
\EndFor

\State \Return $\{\bm{w}_k^\star\}_{k\in [K]}$
\end{algorithmic}
\end{algorithm}

\section{Performance Analysis}
\label{sec:Performance Analysis}

In this section, we provide some theoretical results on the performance of \algo.
In particular, Th. \ref{th:switching} provides motivation to the switching rule in \eqref{eq:loss_stagnation_rule} while Th. \ref{th: theorem_bound} bounds the hypothesis testing performance of \algo. 

\smallskip
\noindent
\textbf{{Assumptions:}} \label{sec:Assumptions}
Before proceeding, let us discuss the assumptions that we adopt in the performance analysis in this section. We adopt standard assumptions on the loss function \cite{li2019convergence}, as follows.
%
\begin{enumerate}[leftmargin=*]
    \item \emph{Smoothness:}
The loss function $\mathcal{L}$ is $L$-smooth. For all $\bm{v}, \bm{w} \in \mathbb{R}^m$,
\begin{equation}
\mathcal{L}(\cdot; \bm{v}) \le \mathcal{L}(\cdot;\bm{w}) + \nabla \mathcal{L}(\cdot;\bm{w})^\top (\bm{v} -  \bm{w}) + \frac{L}{2} \|\bm{v} - \bm{w}\|^2\,,
\end{equation}
\item \emph{Strong Convexity:}
The loss function $\mathcal{L}$ is $\mu$-strongly convex. For all $\bm{v}, \bm{w} \in \mathbb{R}^m$,
\begin{equation}
\mathcal{L}(\cdot; \bm{v}) \ge \mathcal{L}(\cdot;\bm{w}) + \nabla \mathcal{L}(\cdot;\bm{w})^\top (\bm{v} -  \bm{w}) + \frac{\mu}{2} \|\bm{v} - \bm{w}\|^2.
\end{equation}

    \item \emph{Bounded Variance:}
Let $\bm{x}_k^{(i)}$ be sampled uniformly at random from the local dataset $\mathcal{D}_k$. The variance of stochastic gradients is bounded. In formulae, we have that $\forall k = 1, \dots, K$
\begin{equation}
\mathbb{E} \left\| \nabla \mathcal{L}(\bm{x}_k^{(i)};\bm{w}_k^t) - \mathbb{E}( \mathcal{L}(\bm{x}_k^{(i)};\bm{w}_k^t)) \right\|^2 \le \sigma_k^2\,,
\end{equation}
where both expectations are on the distribution $\mathcal{D}_k$.
\item \emph{Bounded Second Moment:}
The expected squared norm of stochastic gradients is uniformly bounded $\forall k = 1, \dots, K$, $t = 1, \dots, T$
\begin{equation}
\mathbb{E} \left\| \nabla \mathcal{L}(\bm{x}_k^{(i)};\bm{w}_k^t) \right\|^2 \le G^2.
\end{equation}
\end{enumerate}


\smallskip
\noindent
\textbf{{On  the Role of the Switching Criterion:}} A central element of \algo is the \emph{switching criterion} between the collaborative and fine-tuning phases. 
Since prematurely switching to fine-tuning may lead to suboptimal personalization and excessive switching may waste communication rounds, we monitor the training loss of the local models. Specifically, at each \ac{bs} $k$, we track the average loss reduction over a sliding window of $T_{\rm W}$ communication rounds: if the improvement falls below a predefined threshold $\epsilon$, the \ac{bs} is marked as ready to local fine-tuning, and once all \acp{bs} are ready, the fine-tuning phase begins. 
Formally, the $k$-th \ac{bs} is marked as ready at time $e_k$ when, for the first time \eqref{eq:loss_stagnation_rule} occurs.

A first question is \emph{when} the switch between Phase 1 and Phase 2 in Algo. \ref{alg:FedLoss}, i.e., the \ac{fedavg} and fine-tuning phases, occurs. 
The next theorem, Th. \ref{th:switching}, shows that the windowed average loss improvement in \eqref{eq:loss_stagnation_rule} decays as $\Ocal(1/e)$.

Let $\bm{w}_{\rm F}^{\star}$ be the optimal federated model, and let $\bm{w}_{\rm F}^{(e)}$ be the average model at each client at communication round $e$, that is 
%
%
\begin{equation} \label{eq:avg_mod}
    \bm{w}_{\rm F}^{(e)} = \frac{1}{K}\sum_{k\in [K]} \bm{w}_k^{(e)}.
\end{equation}
%
A proxy for the switching condition in \eqref{eq:loss_stagnation_rule} is then
\begin{equation} \label{eq: obj_fed_prime}
    \mathbb{E}(\mathcal{L}(\bm{w}_{\rm F}^{(e)},\mathcal{D}_k))- \mathbb{E}(\mathcal{L}(\bm{w}^\star_{\rm F},\mathcal{D}_k)).
\end{equation}
Our next theorem characterized the decay of \eqref{eq: obj_fed_prime} as a function of the communication round $e$.




\begin{thm}[Collaborative-phase optimality gap]
\label{th:switching}
Let Assumptions 1 to 4 hold and $L,\mu,\sigma_k,G$ be defined therein. 
Also, let $\beta$ be the momentum factor and $B_{\rm S}$ the batch size in Algo. \ref{alg:FedLoss}. 
Let us define the non-iid factor 
\ea{
\Gamma = \mathbb{E}(\mathcal{L}(\bm{w}^\star_{\rm F},\mathcal{D}_k))
- \mathbb{E}(\mathcal{L}(\{\bm{w}^\star_k\}_{k\in [K]},\mathcal{D}_k))
\ge 0\,,
}
 where $\mathbb{E}(\mathcal{L}(\{\bm{w}^\star_k\}_{k\in [K]},\mathcal{D}_k))$ is the minimal average loss function attainable by the \acp{bs}; which happens when each \ac{bs} $k$ uses its optimal model $\bm{w}^\star_k$ \footnote{Note that $\Gamma$ measures how much the distributions $\mathcal{D}_k$ are different: as the number of samples grows, then $\Gamma$ goes to zero in case data were i.i.d.}. 
Choose
\ea{
\kappa = \frac{L}{\mu}, \quad 
\gamma = \max\{8\kappa,E\}, \quad 
\eta^{(t)} = \frac{2}{\mu(\gamma+t)\beta}
\label{eq:constant thm 1}
}
for \algo in Algo. \ref{alg:FedLoss}.
Then, at each communication time $e$, we have
\begin{equation} \label{eq:final_theorem}
\begin{split}
    & \mathbb{E}(\mathcal{L}(\bm{w}_{\rm F}^{(e)},\mathcal{D}_k))- \mathbb{E}(\mathcal{L}(\bm{w}^\star_{\rm F},\mathcal{D}_k))\\
& \le
\frac{\kappa}{\gamma+e-1}
\left(
\frac{2B}{\mu}
+ 
\frac{\mu\gamma}{2}
\mathbb{E}(\|\bm{w}_{\rm F}^{(0)}- \bm{w}_{\rm F}^\star\|^2)
\right),
\end{split}
\end{equation}
where
\begin{equation} \label{eq: def_B}
B
=
\frac{1}{K^2}\sum_{k \in [K]} \lb \frac{\sigma_k}{B_{\rm S}} \rb^2
+
6L\Gamma
+
(8(E-1)+(1-\beta^2))G^2.    
\end{equation}

\end{thm}

\begin{IEEEproof}
The proof is provided in App. \ref{app:switching}.
\end{IEEEproof}
Let us comment briefly on the difference between the condition in \eqref{eq:loss_stagnation_rule}
and the condition in \eqref{eq: obj_fed_prime}.
The expectation in \eqref{eq: obj_fed_prime} is taken over the randomness in the stochasticity of the gradient, while the \eqref{eq:loss_stagnation_rule} is evaluated over a set of gradient realizations. 
%

 By the law of large numbers, this empirical average converges almost surely to its expectation as $T_{\rm W}$ grows. Theorem \ref{th:switching} establishes that this expectation decays as $ O(1/e)$, which implies that for any tolerance $\epsilon > 0$ the switching condition is triggered in finite time almost surely.

We also see that for more heterogeneous data, the algorithm switch happens later, as $\Gamma$ increases. A similar behavior is observed for the local optimization steps $E$, yet performing more frequent updates increases the energy cost.

\begin{rem}[Scope of the analysis]
The strongly convex setting of Theorem~\ref{th:switching} is adopted
as a tractable analytical model, rather than as a characterization of
the nonconvex neural-network objective used in Sec.~VI.
Its purpose is to make explicit the dependence of the optimality-gap
bound on data heterogeneity and the number of local updates.
Nonconvex analyses of related distributed momentum methods,
such as~\cite{yu2019momentum}, instead provide stationarity guarantees
in terms of expected squared-gradient norms.
Such results do not directly yield guarantees for the finite-window
switching criterion in~\eqref{eq:loss_stagnation_rule}, which would
require a separate analysis.
We therefore retain the strongly convex setting for analytical
interpretability and assess the switching behavior of the implemented
neural classifier numerically.
\end{rem}

\smallskip
\noindent
\textbf{{On the Need for Personalized Fine-Tuning:}}
\label{sec:personalization_needed}
We now investigate whether the second step of \algo, i.e., fine-tuning, is truly needed. 
To this end, let us focus on the case of $K=2$ and $N=2$ for simplicity, that is, the binary \ac{crc} with two \ac{bs} and two \ac{roo}. Without loss of generality, let us assume Alice is located in \ac{roo} $n=1$, while Trudy is in $n=2$. For this setting, we show a bound in the \ac{det} performance between \algo vs. \ac{fedavg}. 

For the binary \ac{crc} model, the optimal detection model is the \ac{np} test, which from \cite{8798661} is attained by both \algo and \ac{fedavg}. Yet, \ac{fedavg} sees a mixture of data from the two \acp{bs}, thus the \ac{np} test it reaches is suboptimal with respect to the two optimal \ac{np} tests the two \acp{bs} could achieve with enough training data.  Consequently, we can assess the suboptimality of \ac{fedavg} through the \ac{np} test it reaches with respect to the optimal one.
The optimal policy is characterized by a single threshold on the \ac{llr} of the distribution observed at a \ac{bs}.
%
The next theorem shows a bound between the optimal \ac{llr} test as attained by the \ac{np} test versus that attained by \ac{fedavg} as a function of the Total Variation (TV) among the hypotheses. 

\begin{thm}[Binary \ac{crc}: \ac{det} bound of NP vs. \ac{fedavg}]
\label{th: theorem_bound}
Consider the binary \ac{crc} problem and let $p_k^{n}$
be the distribution the \ac{bs} $k$ experiences when the transmitter is in \ac{roo} $\mathcal{R}_n$. Without loss of generality, let us assume Alice is in \ac{roo} $n=1$, thus $p_k^{1}$ are the legitimate distributions while $p_k^{2}$ are the under attack distributions.
%
Next, fix $P_{\rm FA}=\al$  in the \ac{det} curve of \eqref{eq:avg DET} and let
$\tau_{\rm fed}(\alpha)$ and $\tau_{\rm NP}(\alpha)$ 
the threshold for the \ac{llr} test attained by \ac{fedavg} and the \ac{np} test. 
Then the  bound in the \eqref{eq:final_bound} (at the top of the next page) between the corresponding  \ac{md} probability holds where
the \ac{md} probability of the \ac{np} test is denoted as $P_{\rm MD}^\star$ and that of \ac{fedavg} as $P^{\rm fed}_{\rm MD}$.


\end{thm}
\begin{IEEEproof}
The proof is provided in Appendix \ref{app: theorem_bound}.
\end{IEEEproof}

The bound demonstrates the intuition that the higher the mismatch between the distributions, the worse the performance. Interestingly, there is also a dependency on the thresholds $\tau_{\rm fed}$ and $\tau_{\rm NP}$ for the desired \ac{fa}. This can be explained in the following way: with the target region being $n=1$, as $\tau_{\rm fed} \rightarrow 0 $, the acceptance region grows, and \ac{fa} event never occurs. Consequently, the distributions that matter are the attacker ones $p_k^{2}$. A similar reasoning for $\tau_{\rm fed} \rightarrow \infty $: the \ac{md} never occurs, as the acceptance regions shrink, thus the only distributions that matter are the legitimate ones, i.e., $p_k^{1}$.

\begin{figure*}[t]
\centering
\begin{equation} \label{eq:final_bound}
{
P^{\rm fed}_{\rm MD}\le
\tau_{\rm fed} TV\big(p_1^1, p_2^{1} \big) + \max\left(1, \frac{\tau_{\rm fed}}{\tau_{\rm NP}}\right) \cdot \big[TV\big(p_1^2, p_2^{2} \big) +P_{\rm MD}^\star \big]
}
\end{equation}
\end{figure*}

\begin{figure}
    \centering
    \includegraphics[width=0.7\linewidth]{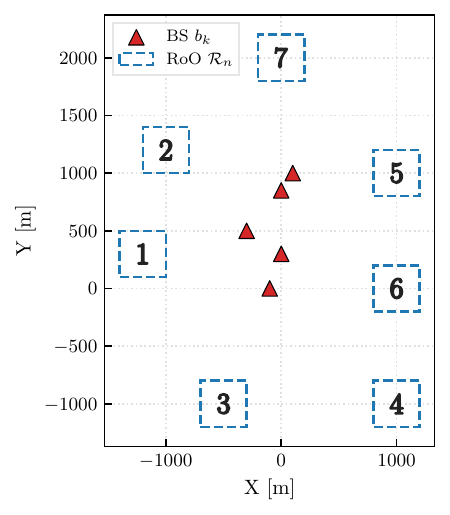}
    \caption{Simulated scenario with $K=5$ \acp{bs} with average distance $D_{\rm bs}=562 \si{\meter}$ and $N=7$ \acp{roo}.}
    \label{fig:sim_scenario}
\end{figure}

%
%


\section{Numerical Results}
\label{sec:setup}

To validate the effectiveness of \algo, we compare it against four baselines: Global, Single, \ac{fedavg}, and FedProx, while considering all possible transmission positions of both Alice and Trudy. 
In the Global setting, a ``virtual'' single \ac{bs} is assumed, with access to a dataset that aggregates the data from all \acp{bs}. This configuration represents the optimal case when the datasets are i.i.d., i.e., when the \acp{bs} are located sufficiently close to each other. Conversely, in the Single setting, each \ac{bs} trains exclusively on its own local dataset, regardless of its size. This approach is expected to perform better when the \acp{bs} are geographically distant, as sharing local information may negatively impact the overall performance. Finally, in the \ac{fedavg} and FedProx settings, the \acp{bs} collaboratively train the model using the \ac{fedavg} and FedProx \cite{nagia2022federated} algorithms, respectively. \remembertext{rev2.7}{
In particular, the FedProx algorithm is a local adaptation strategy where a penalty term is added in the local training to avoid the local model drifting too much from the global model. Moreover, we also study the performance of \algo when a \textbf{fixed} switching time is used, mimicking the behavior of two-stage federated learning algorithms presented in Sec. \ref{sec:intro}. Doing so, we compare the performance of \algo with two-stage federated learning algorithms present in the literature as well as local adaptation strategies. }

The scenario we used for our simulations, depicted in \Cref{fig:sim_scenario} has $K=5$ \acp{bs} that collaborate to classify signals coming from $N=7$ \acp{roo}. The areas are square regions of size $\ell_{\rm S} = 200 \si{\meter}$, while the \acp{bs} are located at random positions, 
\remembertext{rev1.3}{ 
ensuring a nonuniform
BS density and irregular BS placement and verify the generalization capabilities of \algo. The parameter $D_{\rm bs}$ controls the average distance between them, with a default value of $D_{\rm bs} = 563 \si{\meter}$.} We recall that in this paper, we are interested in typical IoT scenarios where nodes are constrained, thus only a few channel samples are available at each \acp{bs}. By default, each \ac{bs} has then available only $16$ data points per class; thus, each $b_k$ has only $ C_k=112$ samples available for training.

The scenario we used for our simulations, depicted in \Cref{fig:sim_scenario} has $K=5$ \acp{bs} that collaborate to classify signals coming from $N=7$ \acp{roo}. The areas are square regions of size $\ell_{\rm S} = 200 \si{\meter}$, while the \acp{bs} are located at random positions, 
\remembertext{rev1.3}{ 
ensuring a nonuniform
BS density and irregular BS placement and verify the generalization capabilities of \algo. The parameter $D_{\rm bs}$ controls the average distance between them, with a default value of $D_{\rm bs} = 563 \si{\meter}$.} We recall that in this paper, we are interested in typical IoT scenarios where nodes are constrained, thus only a few channel samples are available at each \acp{bs}. By default, each \ac{bs} has then available only $16$ data points per class; thus, each $b_k$ has only $ C_k=112$ samples available for training.

\subsection{Channel Model}

We now describe how the channel observations are generated. 
Let $\bm{s}_k\in\mathbb{R}^3$ denote the position of \ac{bs} $k$, and let $\bm{p}\in\mathbb{R}^3$ denote the position of the transmitter. 
The distance between them is $d_k(\bm{p}) = \lVert \bm{s}_k - \bm{p} \rVert$. We adopt a clustered frequency-selective channel model based on 3GPP specifications \cite{etsi2020study}. 
Even under line-of-sight conditions, the received signal consists of $L$ propagation paths. 
The \ac{aoa} of the $\ell$-th path is
$    \xi_{k,\ell}(\bm{p})
    =
    \xi_k(\bm{p}) + \Delta_\ell$, where $\Delta_\ell \sim \mathcal{U}(-\xi_{\mathrm{sp}},\xi_{\mathrm{sp}})$.

We assume \ac{ofdm} signaling with $S$ subcarriers. 
For subcarrier $s\in\{1,\ldots,S\}$, the received signal at \ac{bs} $k$ is
\begin{equation}
    \bm{y}_{k,s}(\bm{p})
    =
    \bm{h}_{k,s}(\bm{p})x
    +
    \bm{n}_{k,s},
\end{equation}
where $    \bm{n}_{k,s}\sim\mathcal{CN}(\bm{0},\sigma_n^2\bm{I})$.

The channel vector is modeled as
\begin{equation}
    \bm{h}_{k,s}(\bm{p})
    =
    \sqrt{\frac{\kappa}{\kappa+1}}
    \bm{h}^{\mathrm{LOS}}_{k,s}(\bm{p})
    +
    \sqrt{\frac{1}{\kappa+1}}
    \bm{h}^{\mathrm{NLOS}}_{k,s}(\bm{p}).
\end{equation}

The LOS component is
\begin{equation}
    \bm{h}^{\mathrm{LOS}}_{k,s}(\bm{p})
    =
    \sum_{\ell=1}^{L}
    \gamma_k(\bm{p})
    \bm{\beta}\!\left(\xi_{k,\ell}(\bm{p})\right)
    e^{-j2\pi s\Delta f\tau_{k,\ell}(\bm{p})},
\end{equation}
where the steering vector is
\begin{equation}
\bm{\beta}(\theta)
=
\frac{1}{\sqrt{M}}
\left[
1,
e^{j\pi\sin\theta},
\ldots,
e^{j\pi(M-1)\sin\theta}
\right]^\mathsf{T},
\end{equation}
and the pathloss coefficient is$
    \gamma_k(\bm{p})
    =
c/(4\pi f_c d_k(\bm{p}))$.

For subcarrier $s\in\{1,\ldots,S\}$, we denote the corresponding complex channel coefficient by
$    \bm{h}_{k,s}(\bm{p}) \in \mathbb{C}^{M}.
$

The \ac{cfr} observed collected at \ac{bs} $k$ is the column vector 
\begin{equation}
    \bm{x}_k(\bm{p})
    =
    \big[
        \bm{h}^T_{k,1}(\bm{p}),
        \bm{h}^T_{k,2}(\bm{p}),
        \ldots,
        \bm{h}^T_{k,S}(\bm{p})
    \big]^\mathsf{T}
    \in \mathbb{C}^{MS}.
\label{eq:channel_observation2}
\end{equation}

Assuming unit-power pilot transmissions, i.e.,
$    \mathbb{E}\{|x|^2\}=1,
$, the resulting \ac{snr} is
\begin{equation}
    \mathrm{SNR}
    =
    \frac{
        \mathbb{E}\!\left\{
            \lVert \bm{h}_{k,s}(\bm{p}) \rVert^2
        \right\}
    }{
        \mathbb{E}\!\left\{
            \lVert \bm{n}_{k,s} \rVert^2
        \right\}
    }.
\end{equation}

The expectation is taken over subcarriers and transmitter/receiver positions. If not otherwise stated, the default parameters we used to generate the channels are the following: a ${\rm SNR} = 10 \, \si{\decibel}$, $\xi_{\mathrm{sp}} = 6 \deg$ and $ c_{\rm DS} = 5 \si{ \nano \second} $, from 3GPP specifications \cite{etsi2020study}.

\begin{table}[t]
\centering
\caption{Summary of channel-model notation.}
\label{tab:channel_notation}
\begin{tabular}{lp{6.1cm}}
\toprule
\textbf{Symbol / Term} & \textbf{Description} \\
\midrule

$\bm{s}_k$ & \textbf{\ac{bs} Position:} Position of \ac{bs} $k$ in $\mathbb{R}^3$ \\

$\bm{p}$ & \textbf{Transmitter Position:} TX position \\

$d_k(\bm{p})$ & \textbf{\ac{bs}--TX Distance:} Distance between \ac{bs} $k$ and TX \\

$\bm{h}_{k,s}(\bm{p})$ & \textbf{SIMO Channel:} channel at $b_k$ for subcarrier $s$\\

$\bm{x}_k(\bm{p})$ & \textbf{\ac{cfr}:} Channel matrix across all subcarriers \\

$\xi_{k,\ell}(\bm{p})$ & \textbf{\ac{aoa}:} AoA of path $\ell$ at \ac{bs} $k$ \\

$\tau_{k,\ell}(\bm{p})$ & \textbf{Path Delay:} Delay of path $\ell$ \\

$\bm{\beta}(\theta)$ & \textbf{Steering Vector:} Array steering vector for AoA $\theta$ \\

$\gamma_k(\bm{p})$ & \textbf{Pathloss Coefficient:} Pathloss between TX and $b_k$ \\

$\kappa$ & \textbf{Rician Factor:} LOS-to-NLOS power ratio \\

$L$ & \textbf{Number of Paths:} Total propagation paths \\

$S$ & \textbf{Number of Subcarriers:} Total OFDM subcarriers \\

$\Delta f$ & \textbf{Subcarrier Spacing:} Frequency spacing \\

$\sigma_n^2$ & \textbf{Noise Variance:} Constant noise variance \\

$\bm{n}_{k,s}$ & \textbf{Noise Vector:} Noise at \ac{bs} $k$, subcarrier $s$ \\

\bottomrule
\end{tabular}
\end{table}


\subsection{Network Architecture}
In this work, we adopt a standard ResNet-18 \ac{cnn} backbone \cite{he2016deep}, augmented with a two-layer fully connected classification head. The input consists of two channels corresponding to the real and imaginary components of the received SIMO-OFDM signal \eqref{eq:channel_observation}, with dimensions $(2, M, S)$. The first convolutional layer is modified to accept two input channels while preserving the original kernel size and stride. The backbone extracts a 512-dimensional feature representation, which is fed into a classifier consisting of a fully connected layer that reduces the dimensionality from 512 to 256, followed by a ReLU activation, a dropout layer with a rate of 0.5, and a final linear layer mapping to $N$ output classes. The use of a well-established architecture ensures reproducibility and facilitates fair comparisons across methods. To make a fair comparison, grid search was employed to find the best hyperparameters across all the methods.
We trained all the methods for $T_{\rm E}=150$ communication rounds, using the \ac{msgd} optimizer and weight decay of $35 \times 10^{-5}$. We used an exponentially decaying learning rate, which is crucial for \algo to not depart from the \ac{fedavg} starting point. Table \ref{tab:hyperparameters} reports the other hyperparameters used for each method. Note that using a ResNet-18 may be computationally heavy for classical federated learning scenarios. However, unlike typical federated learning applications, the clients in our setting are \acp{bs}, which possess significantly stronger computational capabilities. Moreover, as shown in the next section, parameter sharing can be compressed, allowing for substantial bit savings in communication.


\begin{table}[t]
\centering
\caption{Training hyperparameters for different methods.}
\label{tab:hyperparameters}
\begin{tabular}{lccccc} 
\toprule
\textbf{Param.} & \textbf{Global} & \textbf{FedAvg} & \textbf{FedProx} &
\textbf{\algo} &
\textbf{Single}

\\
\midrule

Initial \ac{lr} & $10^{-1}$ & $10^{-2}$ & $10^{-2}$ 
& $10^{-2}$ & $10^{-1}$ 
\\
Batch size & 16 & 32 & 16 & 16 & 16 \\

Prox. term $\mu$ & -- & -- & 0.01 & -- & -- \\
\bottomrule
\end{tabular}
\end{table}

\subsection{Classification Performance}

We evaluate the classification performance of our methods by tracking the average error probability $P_{\rm E}$  across \acp{bs}: the lower, the better. This metric jointly reflects both security and attacker localization performance: the more accurate the classification of transmitting areas is, the more we are able to distinguish Alice from Trudy and, when applicable, to localize the attacker. 

\begin{figure*}[t]
    \centering

    \begin{subfigure}[b]{0.32\textwidth}
        \centering
        \begin{tikzpicture}
\begin{axis}[
     width=5.8cm,
    height=5.5cm,
    xlabel style={font=\scriptsize},
    ylabel style={font=\scriptsize},
    legend style={font=\scriptsize},
    legend style={font=\tiny},
    xlabel={Communication Rounds},
    ylabel={$P_{\rm E}$},
    grid=major,
    legend pos=north east,
]

\addplot[
    color=orange!85!black,
    mark=*,
    mark options={solid},
    very thick
] table[
    col sep=comma,
    x=epochs,
    y expr=\thisrow{fedloss}
] {images/journal/data/accuracy_vs_epochs_fixed.csv};
\addlegendentry{Dynamic Switch}

\addplot[
    color=teal!75!black,
    mark=pentagon*,
    very thick
] table[
    col sep=comma,
    x=epochs,
    y expr=\thisrow{fed_loss_0}
] {images/journal/data/accuracy_vs_epochs_fixed.csv};
\addlegendentry{Fixed: 10 epoch}

\addplot[
    color=green!60!black,
    mark=diamond*,
    very thick
] table[
    col sep=comma,
    x=epochs,
    y expr=\thisrow{fed_loss_1}
] {images/journal/data/accuracy_vs_epochs_fixed.csv};
\addlegendentry{Fixed: 20 epoch}

\addplot[
    color=red!75!black,
    mark=square*,
    very thick
] table[
    col sep=comma,
    x=epochs,
    y expr=\thisrow{fed_loss_2}
] {images/journal/data/accuracy_vs_epochs_fixed.csv};
\addlegendentry{Fixed: 30 epoch}

\addplot[
    color=violet!80!black,
    mark=triangle*,
    very thick
] table[
    col sep=comma,
    x=epochs,
    y expr=\thisrow{fed_loss_3}
] {images/journal/data/accuracy_vs_epochs_fixed.csv};
\addlegendentry{Fixed: 40 epoch}


\end{axis}
\end{tikzpicture}
        \caption{
        Learning curves with fixed switch times versus a dynamic switching policy.}
        \label{fig:learning_vs_fixed_switch}
    \end{subfigure}
    \hfill
    \begin{subfigure}[b]{0.32\textwidth}
        \centering
        \begin{tikzpicture}
\begin{axis}[
     width=5.8cm,
    height=5.5cm,
    xlabel style={font=\scriptsize},
    ylabel style={font=\scriptsize},
    legend style={font=\tiny},
    xlabel={$D_{\rm bs}$ [\si{\meter}]},
    ylabel={$P_{\rm E}$},
    grid=major,
    legend pos=north west,
]

\addplot[
    color=blue!75!black,
    mark=*,
    mark options={solid},
    very thick
] table[
    col sep=comma,
    x=distances,
    y expr=1-\thisrow{global}/100
] {images/journal/data/accuracy_vs_distance.csv};
\addlegendentry{Global}

\addplot[
    color=orange!85!black,
    mark=o,
    very thick
] table[
    col sep=comma,
    x=distances,
    y expr=1-\thisrow{fedloss}/100
] {images/journal/data/accuracy_vs_distance.csv};
\addlegendentry{\algo}

\addplot[
    color=green!60!black,
    mark=diamond*,
    very thick
] table[
    col sep=comma,
    x=distances,
    y expr=1-\thisrow{single}/100
] {images/journal/data/accuracy_vs_distance.csv};
\addlegendentry{Single}

\addplot[
    color=red!75!black,
    mark=square*,
    very thick
] table[
    col sep=comma,
    x=distances,
    y expr=1-\thisrow{fedavg}/100
] {images/journal/data/accuracy_vs_distance.csv};
\addlegendentry{FedAvg \cite{wang2021collaborative,yazdinejad2021federated}}

\addplot[
    color=violet!80!black,
    mark=triangle*,
    very thick
] table[
    col sep=comma,
    x=distances,
    y expr=1-\thisrow{fedprox}/100
] {images/journal/data/accuracy_vs_distance.csv};
\addlegendentry{FedProx \cite{nagia2022federated}}

\end{axis}
\end{tikzpicture}
        \caption{Probability of error $P_{\rm E}$ vs. \acp{bs} distance
        $D_{\rm bs}$.}
        \label{fig:acc_vs_distance}
    \end{subfigure}
    \hfill
    \begin{subfigure}[b]{0.32\textwidth}
        \centering
        \begin{tikzpicture}
\begin{axis}[
     width=5.8cm,
    height=5.5cm,
    xlabel style={font=\scriptsize},
    ylabel style={font=\scriptsize},
    legend style={font=\tiny},
    xlabel={SNR [\si{\decibel}]},
    ylabel={$P_{\rm E}$},
    grid=major,
    legend pos=north east,
]

\addplot[
    color=blue!75!black,
    mark=*,
    mark options={solid},
    very thick
] table[
    col sep=comma,
    x=snr,
    y expr=1-\thisrow{global}/100
] {images/journal/data/accuracy_vs_snr.csv};
\addlegendentry{Global}

\addplot[
    color=orange!85!black,
    mark=o,
    very thick
] table[
    col sep=comma,
    x=snr,
    y expr=1-\thisrow{fedloss}/100
] {images/journal/data/accuracy_vs_snr.csv};
\addlegendentry{\algo}

\addplot[
    color=green!60!black,
    mark=diamond*,
    very thick
] table[
    col sep=comma,
    x=snr,
    y expr=1-\thisrow{single}/100
] {images/journal/data/accuracy_vs_snr.csv};
\addlegendentry{Single}

\addplot[
    color=red!75!black,
    mark=square*,
    very thick
] table[
    col sep=comma,
    x=snr,
    y expr=1-\thisrow{fedavg}/100
] {images/journal/data/accuracy_vs_snr.csv};
\addlegendentry{FedAvg \cite{wang2021collaborative,yazdinejad2021federated}}

\addplot[
    color=violet!80!black,
    mark=triangle*,
    very thick
] table[
    col sep=comma,
    x=snr,
    y expr=1-\thisrow{fedprox}/100
] {images/journal/data/accuracy_vs_snr.csv};
\addlegendentry{FedProx \cite{nagia2022federated}}

\end{axis}
\end{tikzpicture}
        \caption{Probability of error vs. \ac{snr}.}
        \label{fig:acc_vs_snr}
    \end{subfigure}

    \caption{Learning performance and probability of error under different
    system conditions.}
    \label{fig:performance_comparison}
\end{figure*}
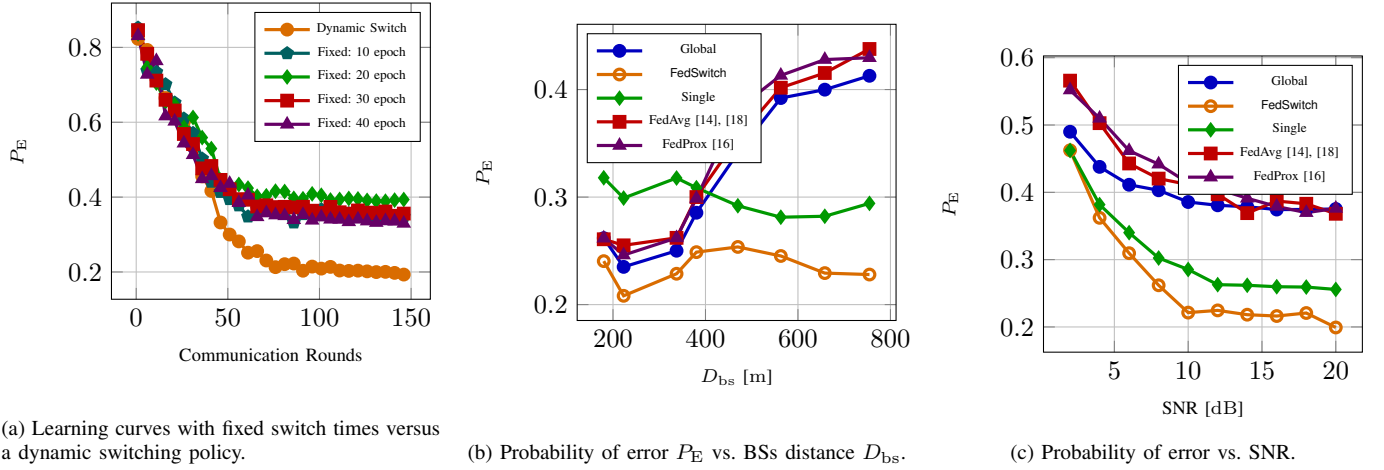

\remembertext{rev2.1}{

\smallskip
\noindent
\textbf{\emph{Error Probability With Fixed Switch Time:}} 
A novelty of the paper, besides the theoretical framework, is the dynamic switching time. To show its benefit, we see from Fig.~\ref{fig:learning_vs_fixed_switch} the performance of \algo compared to a fixed switch approach. We notice that the dynamic switch, happening around the 15th communication round, is beneficial for training and attains the lowest $P_{\rm E}$.
}

\smallskip
\noindent
\textbf{\emph{Error Probability vs. \ac{snr}:}} 
From \Cref{fig:acc_vs_snr}, we observe that increasing \ac{snr} consistently improves the overall performance across all methods, as expected. In this setting, \algo achieves the best performance among all considered approaches. We also note that, despite the limited number of training samples, \ac{fedavg} and FedProx reach performance comparable to the Global model at high \ac{snr} values, showing the effectiveness of federated training.

\smallskip
\noindent
\textbf{\emph{Error Probability vs. \acp{bs} Distance:}} Fig.~\ref{fig:acc_vs_distance} shows the error probability for different \acp{bs}' average distance $D_{\rm bs}$. We notice that the performance of Global, FedAvg, and FedProx collapses as distance increases, as the data across \acp{bs} becomes more non-i.i.d. This high data-heterogeneity is not captured by FedProx, which underperforms with respect to \algo that, thanks to the adaptive finetune phase, is more robust across the methods and achieves the best overall performance.

\begin{figure}
    \centering
\pgfplotstableread[col sep=comma]{images/journal/data/fl_bar_results.csv}\datatable

\begin{tikzpicture}
\begin{axis}[
    width=8cm,
    height=5cm,
    xlabel={Number of Samples},
    ylabel={$P_{\rm E}$},
    legend style={font=\scriptsize},
    grid=major,
    ybar,
    bar width=6pt,
    enlarge x limits=0.15,
    xtick=data,
    xticklabels from table={\datatable}{samples},
    legend pos=north east,
    ymin=0,
]

\addplot[
    color=blue!75!black,
    fill=blue!75!black,
    very thick
] table[
    col sep=comma,
    x expr=\coordindex,
    y expr=1-\thisrow{Global}/100
] {images/journal/data/fl_bar_results.csv};
\addlegendentry{Global}

\addplot[
    color=orange!85!black,
    fill=orange!85!black,
    very thick
] table[
    col sep=comma,
    x expr=\coordindex,
    y expr=1-\thisrow{FedSwitch}/100
] {images/journal/data/fl_bar_results.csv};
\addlegendentry{\algo}

\addplot[
    color=green!60!black,
    fill=green!60!black,
    very thick
] table[
    col sep=comma,
    x expr=\coordindex,
    y expr=1-\thisrow{Single}/100
] {images/journal/data/fl_bar_results.csv};
\addlegendentry{Single}

\addplot[
    color=red!75!black,
    fill=red!75!black,
    very thick
] table[
    col sep=comma,
    x expr=\coordindex,
    y expr=1-\thisrow{FedAvg}/100
] {images/journal/data/fl_bar_results.csv};
\addlegendentry{FedAvg \cite{wang2021collaborative,yazdinejad2021federated}}

\addplot[
    color=violet!80!black,
    fill=violet!80!black,
    very thick
] table[
    col sep=comma,
    x expr=\coordindex,
    y expr=1-\thisrow{FedProx}/100
] {images/journal/data/fl_bar_results.csv};
\addlegendentry{FedProx \cite{nagia2022federated}}

\end{axis}
\end{tikzpicture}
    \caption{Probability of error $P_{\rm E}$ vs. number of training samples $C_k$, equal for all \acp{bs}.}
    \label{fig:acc_vs_samples}
\end{figure}
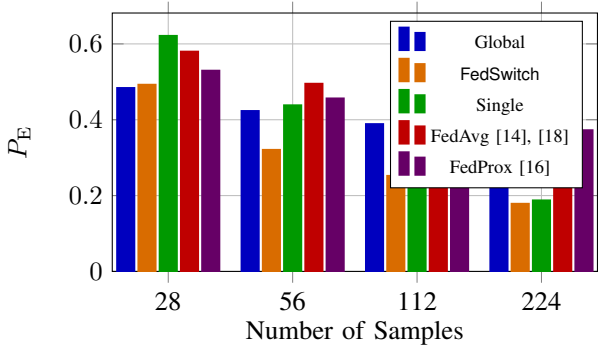

\smallskip
\noindent
\textbf{\emph{Error Probability vs. Number of Samples:}} 
 \Cref{fig:acc_vs_samples} reports the average error probabilit $P_{\rm E}$ as a function of the total number of training samples $C_k$, same for all \acp{bs}. For small sample sizes (e.g., $C_k=28$ corresponds to 4 samples per class), the performance is primarily limited by data scarcity; in this regime, Global and \algo achieve the best results, while the Single approach performs poorly due to insufficient local data. As the number of samples increases, the bottleneck shifts from data scarcity to data heterogeneity. Although overall performance improves, the Single approach becomes increasingly competitive and eventually outperforms the single-optimum methods (i.e., Global and \ac{fedavg}). Across all regimes, \algo consistently achieves the best (or almost best) performance. In the subsequent simulations, we use $C_k=112$ samples for training, as in this regime, we have data scarcity yet enough data to have information sharing across \acp{bs}.

\begin{figure}
    \centering
    \begin{tikzpicture}
\begin{axis}[
    width=8cm,
    height=5cm,
    xlabel={Communication Rounds},
    legend style={font=\scriptsize},
    ylabel={$P_{\rm E}$},
    grid=major,
    legend pos=north east,
]

\addplot[
    color=blue!75!black,
    mark=*,
    very thick
] table[
    col sep=comma,
    x=epoch,
    y expr=1-\thisrow{fedloss_0}/100
] {images/journal/data/accuracy_vs_compression.csv};
\addlegendentry{$\mathrm{R}_k$=0.1}

\addplot[
    color=orange!85!black,
    mark=o,
    very thick
] table[
    col sep=comma,
    x=epoch,
    y expr=1-\thisrow{fedloss_1}/100
] {images/journal/data/accuracy_vs_compression.csv};
\addlegendentry{$\mathrm{R}_k$ = 0.3}

\addplot[
    color=green!60!black,
    mark=diamond*,
    very thick
] table[
    col sep=comma,
    x=epoch,
    y expr=1-\thisrow{fedloss_2}/100
] {images/journal/data/accuracy_vs_compression.csv};
\addlegendentry{$\mathrm{R}_k$= 1.0 (no compression)}

\end{axis}
\end{tikzpicture}
    \caption{Learning Curves vs. compression factor $\mathrm{R}_k$.}
    \label{fig:learning_vs_compression}
\end{figure}
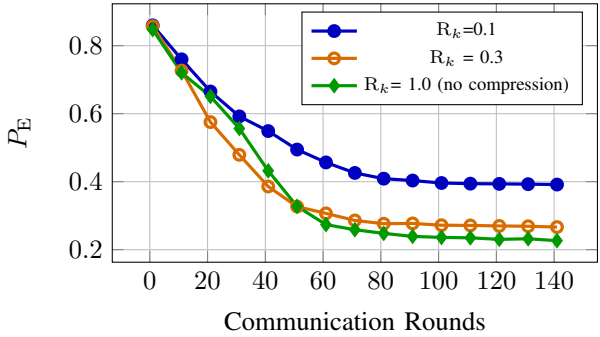

\smallskip
\noindent
\textbf{\emph{Learning curves vs. Compression Factor $\mathrm{R}_k$}} 
Fig.~\ref{fig:learning_vs_compression} shows the learning curves of \algo under different compression factors $\mathrm{R}_k$. To obtain these curves, we applied top-$k$ sparsification, where for each layer $\ell$, only the local weights satisfying 
\begin{equation}
    |w_\ell| \ge P_{\mathrm{R}_k}(|\bm{w}_\ell|)
\end{equation}
are retained and transmitted to the central server. Here, $P_{\mathrm{R}_k}(\cdot)$ denotes the $\mathrm{R}_k$-th empirical quantile of the absolute weights in layer $\ell$. In other words, we transmit only the $\mathrm{R}_k$ percentage of top weights in terms of magnitude, and the remaining $(1 - R_k)$  percent of the weights is masked to zero. 

During the aggregation phase, the central server computes the element-wise average of these sparsified local models. Because each client independently zeroes out different parameters based on their local training, the global model update is formed by averaging only the most significant parameters of the \acp{bs}. This mechanism drastically reduces the communication bandwidth required per round while preserving the most salient features necessary for convergence. As \Cref{fig:learning_vs_compression} shows, \algo maintains strong robustness even under significant parameter compression. This behavior suggests that the learned representations and optimization dynamics of \algo are intrinsically resilient to perturbations and reduced numerical precision. 
We also observe that the compressed model with $\mathrm{R}_k=0.3$ can learn slightly faster. This is because having fewer weights effectively results in lighter models, which require fewer iterations to converge. However, at convergence, the maximum performance remains lower than in the no-compression case $\mathrm{R}_k=1$.


\subsection{Security Performance}

\begin{figure*}[h]

    \begin{subfigure}[t]{0.32\textwidth}
        \centering
        \begin{tikzpicture}
\begin{axis}[
    width=5.8cm,
    height=5.5cm,
    xlabel style={font=\scriptsize},
    ylabel style={font=\scriptsize},
    legend style={font=\scriptsize},
    xmode=log,
    ymode=log,
    xlabel={$\bar{P}_{\rm FA}$},
    ylabel={$\bar{P}_{\rm MD}$},
    grid=major,
    legend pos=south west,
    xmin=1e-2,
    xmax=1,
    ymin=1e-2,
    ymax=1,
]

\addplot[
    color=blue!75!black,
    mark=*,
    very thick
] table[
    x=avg_pfa_global_curve,
    y=avg_pmd_global_curve,
    col sep=comma
] {images/journal/data/roc_curves_dist_10.csv};
\addlegendentry{Global}

\addplot[
    color=orange!85!black,
    mark=o,
    very thick
] table[
    x=avg_pfa_fed_curve,
    y=avg_pmd_fed_curve,
    col sep=comma
] {images/journal/data/roc_curves_dist_10.csv};
\addlegendentry{\algo}

\addplot[
    color=green!60!black,
    mark=diamond*,
    very thick
] table[
    x=avg_pfa_single_curve,
    y=avg_pmd_single_curve,
    col sep=comma
] {images/journal/data/roc_curves_dist_10.csv};
\addlegendentry{Single}

\addplot[
    color=red!75!black,
    mark=square*,
    very thick
] table[
    x=avg_pfa_fed_avg_curve,
    y=avg_pmd_fed_avg_curve,
    col sep=comma
] {images/journal/data/roc_curves_dist_10.csv};
\addlegendentry{FedAvg \cite{wang2021collaborative,yazdinejad2021federated}}

\addplot[
    color=violet!80!black,
    mark=triangle*,
    very thick
] table[
    x=avg_pfa_fed_prox_curve,
    y=avg_pmd_fed_prox_curve,
    col sep=comma
] {images/journal/data/roc_curves_dist_10.csv};
\addlegendentry{FedProx \cite{nagia2022federated}}

\addplot[
    gray,
    dashed,
    thick
] coordinates {(1e-2,1e-2) (1,1)};

\end{axis}
\end{tikzpicture}
        \caption{\ac{det} curve with average \acp{bs} distance $D_{\rm bs} = 220\,\si{\meter}$.}
        \label{fig:DET_10}
    \end{subfigure}
    \hfill
    \begin{subfigure}[t]{0.32\textwidth}
        \centering
        \begin{tikzpicture}
\begin{axis}[
    width=5.8cm,
    height=5.5cm,
    legend style={font=\scriptsize},
    xlabel style={font=\scriptsize},
    ylabel style={font=\scriptsize},
    xmode=log,
    ymode=log,
    xlabel={$\bar{P}_{\rm FA}$},
    ylabel={$\bar{P}_{\rm MD}$},
    grid=major,
    legend pos=south west,
    xmin=1e-2,
    xmax=1,
    ymin=1e-2,
    ymax=1,
]

\addplot[
    color=blue!75!black,
    mark=*,
    very thick
] table[
    x=avg_pfa_global_curve,
    y=avg_pmd_global_curve,
    col sep=comma
] {images/journal/data/roc_curves_dist_50.csv};
\addlegendentry{Global}

\addplot[
    color=orange!85!black,
    mark=o,
    very thick
] table[
    x=avg_pfa_fed_curve,
    y=avg_pmd_fed_curve,
    col sep=comma
] {images/journal/data/roc_curves_dist_50.csv};
\addlegendentry{\algo}

\addplot[
    color=green!60!black,
    mark=diamond*,
    very thick
] table[
    x=avg_pfa_single_curve,
    y=avg_pmd_single_curve,
    col sep=comma
] {images/journal/data/roc_curves_dist_50.csv};
\addlegendentry{Single}

\addplot[
    color=red!75!black,
    mark=square*,
    very thick
] table[
    x=avg_pfa_fed_avg_curve,
    y=avg_pmd_fed_avg_curve,
    col sep=comma
] {images/journal/data/roc_curves_dist_50.csv};
\addlegendentry{FedAvg \cite{wang2021collaborative,yazdinejad2021federated}}

\addplot[
    color=violet!80!black,
    mark=triangle*,
    very thick
] table[
    x=avg_pfa_fed_prox_curve,
    y=avg_pmd_fed_prox_curve,
    col sep=comma
] {images/journal/data/roc_curves_dist_50.csv};
\addlegendentry{FedProx \cite{nagia2022federated}}

\addplot[
    gray,
    dashed,
    thick
] coordinates {(1e-2,1e-2) (1,1)};

\end{axis}
\end{tikzpicture}
        \caption{\ac{det} curve with average \acp{bs} distance $D_{\rm bs} = 563\,\si{\meter}$.}
        \label{fig:DET_50}
    \end{subfigure}
    \hfill
    \begin{subfigure}[t]{0.32\textwidth}
        \centering
        \begin{tikzpicture}

\begin{axis}[
    width=5.8cm,
    height=5.5cm,
    xlabel style={font=\scriptsize},
    ylabel style={font=\scriptsize},
    legend style={font=\scriptsize},
    xlabel={$P_{\rm FA}$},
    ylabel={$P_{\rm MD}$},
    grid=both,
    legend style={
        at={(0.02,0.02)},
        anchor=south west,
        fill=white,
        draw=black
    },
    legend image post style={
        only marks,
        mark options={black}
    },
    xmin=0.1,
    xmax=0.9,
    xmode = log,
    ymode = log,
    ymin=0,
    enlargelimits=false,
    line width=1pt,
    mark size=3pt
]


\pgfplotstableread[col sep=comma]{images/journal/data/bound_verifier_sigma_2.csv}\datatableA
\pgfplotstableread[col sep=comma]{images/journal/data/bound_verifier_sigma_1.1.csv}\datatableB


\addplot[
    blue,
    solid,
    mark=x
]
table[
    x=pfa,
    y={PMD Optimal}
] {\datatableA};

\addlegendentry{$P^\star_{\rm md}$: \ac{np} Local}

\addplot[
    blue,
    solid,
    mark=square
]
table[
    x=pfa,
    y={PMD Effective}
] {\datatableA};

\addlegendentry{$P^{\rm fed}_{\rm MD}$: \ac{np} Global}

\addplot[
    blue,
    solid,
    mark=o
]
table[
    x=pfa,
    y={Theoretical Bound}
] {\datatableA};

\addlegendentry{\Ac{fedavg} Bound}


\addplot[
    red,
    dashed,
    mark=x
]
table[
    x=pfa,
    y={PMD Optimal}
] {\datatableB};

\addplot[
    red,
    dashed,
    mark=square
]
table[
    x=pfa,
    y={PMD Effective}
] {\datatableB};

\addplot[
    red,
    dashed,
    mark = o
]
table[
    x=pfa,
    y={Theoretical Bound}
] {\datatableB};

\end{axis}
\end{tikzpicture}
        \caption{Verification of bound \eqref{eq:final_bound} under Gaussian distribution mismatch for $\sigma = 1$. The solid line has $\sigma_{\rm G} = 2$ while the dashed line has $\sigma_{\rm G} = 1.1$.}
        \label{fig:DET_verifier}
    \end{subfigure}
    
    \caption{Evaluation of the proposed detection method and theoretical bound.}
    \label{fig:DET_comparison}
\end{figure*}
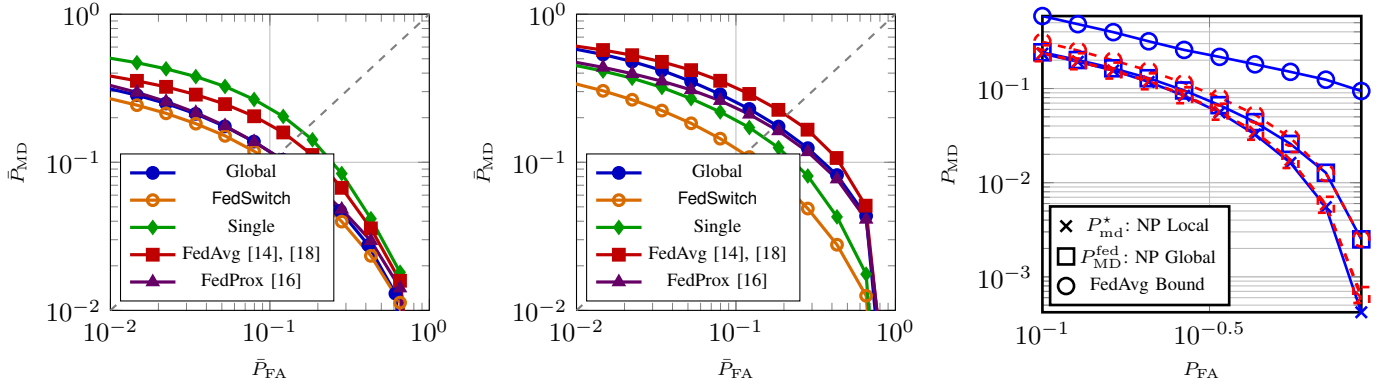

In this section, we evaluate security performance in terms of average \ac{md} $\bar{P}_{\rm MD}$ and \ac{fa} $\bar{P}_{\rm FA}$ probabilities defined in \Cref{sub: perf_eval} via \ac{det} curves, which show the tradeoff between the two. In fact, the greater the \ac{fa}, the lower the \ac{md}, as expected. From \Cref{fig:DET_50}, we observe that the Global, \ac{fedavg}, and FedProx approaches are negatively impacted by data heterogeneity as $D_{\rm bs} = 563 \si{\meter}$,
which degrades detection performance. The Single approach struggles as well due to the limited amount of local training data, which restricts its ability to generalize. In contrast, \algo consistently achieves best performance, benefiting from the joint exploitation of local and global information, which effectively mitigates the effects of heterogeneity while improving the \ac{md}/\ac{fa} trade-off. In \Cref{fig:DET_10}, we show the \ac{det} computed at a different distance $D_{\rm bs}=220\,\si{\meter}$. 
Comparing it with \Cref{fig:DET_50}, we observe that the performance gap between the Single and the aggregated algorithms is reversed. 
This behavior arises because, when the \acp{bs} are closer to the transmitter, the observed data become more i.i.d., which improves the effectiveness of information sharing across \acp{bs}. Notably, \algo achieves the best performance even in this setting.


\smallskip
\noindent
\textbf{{Theoretical Bound Verification:}} To validate the tightness and correctness of the theoretical performance bound derived in Theorem \ref{th: theorem_bound}, we simulate a binary hypothesis testing problem using synthetic Gaussian data, for which the variational distances are computable in closed form as the underlying statistic distributions are known. We show results for $b_1$ assuming Alice is in region $n=1$ without loss of generality. The \ac{llr} at $b_1$ reaches is the pooled test
\begin{equation} \label{eq: pooled_np}
   \Lambda_{\rm G}= \frac{p_1^{1}+p_2^{1}}{p_1^{2}+p_2^{2}} = \frac{p^1}{p^2},
\end{equation}
as shown in Appendix B, which contains the mixture of the distributions. The optimal \ac{np} test is the local one, defined as
\begin{equation}
   \Lambda_{\rm L}= \frac{p_1^{1}}{p_1^{2}}\,.
\end{equation}

Assuming, as an example, the same under attack distributions we have that $p_1^{2} =p^{2}$ for both \acp{bs}, while we model the legitimate distributions as $p_1^{1} \sim \mathcal{N}(0, \sigma^2)$ and $p^1 \sim \mathcal{N}(0, \sigma_{\rm G}^2) $. This is an example where the \acp{bs} have the same under attack distributions while the legitimate ones are Gaussians with different variances. As shown in \Cref{fig:DET_verifier}, we have an empirical verification of the proposed bound in \eqref{eq:final_bound} under Gaussian distribution mismatch. In the simulation, we used $\sigma = 1$, the solid line has $\sigma_{\rm G} = 2$ while the dashed line has $\sigma_{\rm G} = 1.1$. Across all simulated target $P_{\rm FA}$ values, the empirically observed \ac{md} probability $P^{\rm fed}_{\rm MD}$ attained using \eqref{eq: pooled_np} follows the profile of the optimal curve   $P^\star_{\rm MD}$ and remains strictly below the analytical bound, confirming its validity. Moreover, we see that the greater the mismatch in the distributions, the worse the performance, as expected and correctly captured by the bound.

\section{Conclusions} \label{sec: conclusion}
In this paper, we presented \algo, a novel \ac{fl} framework designed to mitigate the challenges of non-IID datasets through adaptive local fine-tuning. By monitoring the training loss, \algo dynamically determines the optimal switching time between the federated aggregation and local fine-tuning phases. We provided a theoretical analysis demonstrating the structural limitations of \ac{fedavg} under non-IID conditions and proved how \algo successfully mitigates them. Numerical evaluations utilizing realistic, 3GPP-compliant channel models confirm that \algo outperforms both state-of-the-art federated and conventional non-federated baselines in terms of error probability.

\bibliographystyle{IEEEtran}
\bibliography{biblio}

\appendices

\section{Proof of Theorem \ref{th:switching}}
\label{app:switching}
Let us define $\{t \}$ as the index indicating the $t$-th step of \ac{msgd}, and $\mathcal{I}_E  \subseteq \{t \}$ as the steps where a communication round occurs, which happens every $E$ gradient steps. Thus we have that at step $t$-th step, at the \ac{bs} $k$ the model is
\begin{equation}
\bm{w}^{(t+1)}_{k}=\begin{cases}
\bm{w}^{(t)}_{k}
-
\eta^{(t)} \bm{m}_k^{(t)},
\qquad \text{if}\,\, t \notin \mathcal{I}_E
\\[0.8em]

 \frac{1}{K} \sum_{k \in [K]} \bm{w}^{(t)}_{k} 
-
\eta^{(t)} \bm{m}_k^{(t)} \,\text{otherwise}\,,
\end{cases}    
\end{equation}
where
\begin{equation}
    \bm{m}_k^{(t)} =  \beta \bm{g}_k^{(t)} + (1-\beta) \bm{m}_k^{(t-1)}\,,
\end{equation}
is the moment update.
Following the standard analysis of \ac{msgd} convergence, we define
\eas{
\bm{m}^{(t)} & =\frac{1}{K}\sum_{k \in [K]} \bm{m}_{t}^{k} \\
\bm{g}^{(t)} & = \frac{1}{K} \sum_{k\in [K]}\nabla \mathcal{L}(\bm{w}_k^t; \bm{H}^{(k,i)})  \\
\bm{w}^{(t)} & =\frac{1}{K}\sum_{k \in [K]}\bm{w}^{(t)}_{k}
}{\label{eq:averages}}
%
%
With the definitions in \eqref{eq:averages}, if $t \in \mathcal{I}_E$, we unfold $\bm{w}^{(t+1)}$ as
%
\begin{equation} \label{eq: dim_first}
\begin{split}
    \mathbb{E} &(\lVert \bm{w}^{(t+1)} -\bm{w}^{\star}_{\rm F} \lVert^2 )= \mathbb{E}(\lVert \bm{w}^{(t)}-\eta^{(t)} \bm{m}^{(t)}-\bm{w}^{\star}_{\rm F} \lVert^2 ) \\
    &= \mathbb{E}(\lVert \bm{w}^{(t)}-\eta^{(t)}( \beta \bm{g}^{(t)} + (1-\beta) \bm{m}^{(t-1)})-\bm{w}^{\star}_{\rm F}\lVert^2) \\
    & \leq \underbrace{\mathbb{E}(\lVert \bm{w}^{(t)}-\eta^{(t)} \beta \bm{g}^{(t)} -\bm{w}^{\star}_{\rm F}\lVert^2}_{\rm A} ) +  \underbrace{\mathbb{E}(\lVert \eta^{(t)} (1-\beta) \bm{m}^{(t-1)}\lVert^2}_{\rm C}) \\
\end{split}
\end{equation}
Let us define $\widetilde{\eta}^{(t)} = \beta \eta^{(t)}$: if $\widetilde{\eta}^{(t)} \le \frac{1}{4L}$. 

Assuming also $\widetilde{\eta}^{(t)}$ is non-increasing and $
\widetilde{\eta}^{(t)} \le 2\widetilde{\eta}^{(t+E)}
$,
for all $t \ge 0$, the inequality A is bounded from \cite{li2019convergence}, recalling that batch-training reduces the variance of the stochastic gradients by a factor of $B_{\rm s}^2$, as 

\begin{equation} \label{eq: bound A}
\begin{split}
\text{A}
\le
(1-\widetilde{\eta}^{(t)}\mu)\mathbb{E}\|\bm{w}^{(t)}-\bm{w}_{\rm F}^\star\|^2
+
(\widetilde{\eta}^{(t)})^{2} \frac{1}{K^2}\sum_{k=1}^{N}
 \sigma_k^2 \\
+
6L(\widetilde{\eta}^{(t)})^{2}\Gamma
+
8 (\widetilde{\eta}^{(t)})^{2}(E-1)^2G^2
\end{split}
\end{equation}

The quantity $\rm C$ in \eqref{eq: dim_first} is bounded using the next lemma, Lem. \ref{lemma: bound_c} as 
\begin{equation} \label{eq: bound C}
    {\rm{C}} \leq (\eta^{(t)})^2(1-\beta)^2G^2.
\end{equation}

\begin{lemma} \label{lemma: bound_c}
Under the conditions in the statement of Th. \ref{th:switching}, it holds that

    \begin{equation}
        \mathbb{E}(\lVert \eta^{(t)} (1-\beta) \bm{m}^{(t-1)}\lVert^2)\leq (\eta^{(t)})^2(1-\beta)^2G^2.
    \end{equation}
\end{lemma}
\begin{IEEEproof}
The proof is by induction. For $t \leq 0$, the update $\bm{m}^{(t)}$ is zero. For $t=1$ we have that $m_1 = \beta g_1$, and from Assumption 4, its expected value is bounded by $\beta^2G^2$. Let us assume \eqref{eq: bound C} is valid for step $t-1$, we now demonstrate this implies it is valid for step $t$.

\begin{equation}
\begin{split}
    &\mathbb{E}(\lVert \eta^{(t)} (1-\beta) \bm{m}^{(t-1)}\lVert^2) \\  
    &=(\eta^{(t)})^2(1-\beta)^2\mathbb{E}(\lVert \bm{m}^{(t-1)}\lVert^2) \\
    & \leq (\eta^{(t)})^2(1-\beta)^2 \mathbb{E}(\beta \lVert \bm{g}^{(t-1)}+(1-\beta) \bm{m}^{(t-2)} \lVert^2)\\
    & \leq (\eta^{(t)})^2(1-\beta)^2 [\mathbb{E}( \lVert  \beta \bm{g}^{(t-1)}\lVert^2) + \mathbb{E}(\lVert (1-\beta) \bm{m}^{(t-2)}\lVert^2)] \\
    & \leq (\eta^{(t)})^2(1-\beta)^2[\beta^2G^2 + (1-\beta)^2 G^2] \leq (\eta^{(t)})^2(1-\beta)^2 G^2.
    \end{split}
\end{equation}
\end{IEEEproof}
Using the definition of $B$ in \eqref{eq: def_B}
and combining \eqref{eq: bound A} and \eqref{eq: bound C}, we have that \eqref{eq: dim_first} can be written as
\begin{equation}
    \mathbb{E}(\lVert \bm{w}^{(t+1)} -\bm{w}^{\star}_{\rm F} \lVert^2 ) \leq (1-\widetilde{\eta}^{(t)}\mu)\mathbb{E}(\|\bm{w}^{(t)}-\bm{w}^{\star}_{\rm F}\|^2 ) + (\widetilde{\eta}^{(t)})^{2} B\,,
\end{equation}
which has the same form as in \cite{li2019convergence}.
Consequently, with $\kappa = \frac{L}{\mu}$, $\gamma = \max\{8\kappa, E\}$ by fixing $\widetilde{\eta}^{(t)} = \frac{2}{\mu(\gamma + t)}$ we attain from \cite{li2019convergence} that at each communication time $e$, we have our theorem \eqref{eq:final_theorem}.

\section{Proof of Theorem \ref{th: theorem_bound}}
\label{app: theorem_bound}
%
%
The data distribution observed over the datasets of the \ac{bs} $k=\{1,2\}$ is
\eas{
\Dcal_k = \f 1 2 \lb p_{k}^{1}+p_{k}^{2} \rb
}{\label{eq:datasetndistributino}}
%
%


%
Let $\bm{w}_{\rm F}^{\star}$ and $\bm{w}_{\rm G}^{\star}$ be the optimal federated and global models, where the global model is trained on the union of the \acp{bs} data. The optimal \ac{fedavg} coincides with global training, in fact
\begin{equation} \label{eq: conv opt}
\begin{split}
            \bm{w}_{\rm F}^{\star} 
            &= \arg\min_{\bm{w}} \frac{1}{K}\sum_{k\in [K]} \,\mathcal{L}\bigl(\mathcal{D}_k; \bm{w}\bigr) \\
            &= \arg\min_{\bm{w}} \mathcal{L} \lb \bigcup_{k\in [K]} \mathcal{D}_k;\bm{w} \rb = \bm{w}_{\rm G}^{\star}\,.
\end{split}
\end{equation}
%
%
%
Under the  assumptions in Sec. \ref{sec:Assumptions} we have that \ac{fedavg} achieves the optimal LLR test  $\Lambda_{\rm fed}$ over the global data distribution in \eqref{eq: conv opt}, that is
%
\begin{equation} \label{eq:pooled_test}
\Lambda_{\rm G}=
%
\f {p_{1}^{1}+p_{2}^{1}}{p_{1}^{2}+p_{2}^{2}}=\f {p^1} {p^2} \,.
\end{equation}
For a given \ac{fa}, let us define the MD probability $P_{\rm \ac{md}}$. Let $\delta_{\rm fed}$ be the acceptance region of the null hypothesis in the global model. 
On the other hand, the NP test attains the optimal test at each \ac{bs} as $
\Lambda_{{\rm NP}, k} = p_k^{1} /p_k^{2}$.
Accordingly, let  $\delta_{{\rm NP}, k}$ the acceptance region of the  hypothesis $n=1$  \ac{bs} $k$.


%

In the following, let's focus on the case $k=1$, that is on the performance of the \ac{bs} $1$.
The derivation is symmetric for the case $k=2$.
The MD probabilities in the global model and the personalized models are obtained as
\eas{
    P^{\rm fed}_{\rm MD}& =P \lb x \in \delta_{\rm fed} | x \sim p_1^2 \rb \\
    P_{\rm MD}^\star& =P \lb x \in \delta_{\rm NP} | x \sim p_1^2 \rb \,,
}{\label{eq:opt_pmd}}
i.e., the probability that samples belonging to \ac{roo} 2 for \ac{bs} 1 are classified as coming from \ac{roo} 1.
%

\noindent
\textbf{Bounding the Performance Gap via Total Variation:}
%
 In general, having a constraint on the \ac{fa} and minimizing the \ac{md} is equivalent to minimizing the uncostrained Lagrangian objective (with $k=\{\rm fed,NP\}$)
\begin{equation} \label{eq:objective_function}
Q(\delta_k,\tilde{\tau}_k) = \tilde{\tau}_k P_{\rm MD}(\delta_k) + P_{\rm \ac{fa}}(\delta_k)\,, 
\end{equation}
where $\delta_k$ is the {decision region} under the hypothesis $H_0$, and $\tilde{\tau}_k=\frac{1}{\tau_k}$  is the reciprocal of the threshold of the \ac{np} test for a desired \ac{fa} $\tau_k$. Thus, both fed and NP minimize \eqref{eq:objective_function} according to \emph{their data distributions}, thus they will end up with different regions and thresholds to match the \ac{fa} and then minimize the \ac{md}.  

Note that we can compactly write \eqref{eq:objective_function} recalling that the distributions that client 1 experiences are $ p^2_1(x)$ and $ p_1^1(x)$ as 
\begin{equation} \label{eq:obj_expanded}
    Q(\delta_k,\tilde{\tau}_k)=1+\int_{\delta_k} \tilde{\tau}_k  p^2_1(x)-p_1^1(x) \, \,.
\end{equation}

We want to compute the optimality gap between the performance of fed and NP $
    \Delta = P_{\rm MD}(\delta_{\rm fed}) - P_{\rm MD}(\delta_{\rm NP})$.
Let $f_{\rm NP}(x) = p_1^1(x) - \tilde{\tau}_{\rm NP} p_1^2(x) $
 be the optimal decision function for the distributions $(p_1^1(x), p^2_1(x))$ which in fact induces the region that minimizes \eqref{eq:obj_expanded} as it is designed for the "correct distributions" 
 \begin{equation}
 \delta_{\rm NP} = \lcb x : \frac{p_1^1(x)}{p_1^2(x)} \ge \tilde{\tau}_{\rm NP} \iff f_{\rm NP}(x)  \ge 0\rcb
 \end{equation}
 Similarly, let us define $f_{\rm fed}(x) = p^1(x) - \tilde{\tau}_{\rm fed} p^2(x)$ as the federated decision function which induces the region 
\begin{equation}
     \delta_{\rm fed} = \lcb x : \frac{p^{1}(x)}{p^{2}(x)} \ge \tilde{\tau}_{\rm fed} \iff f_{\rm fed}(x)  \ge 0\rcb \,.  
\end{equation}

\subsubsection*{1. Decomposition via Symmetric Difference}

Let us focus first on the difference between the objectives of the models, recalling that by hypothesis they induce the same \ac{fa}, thus ($P_{\rm \ac{fa}}, (\delta_{\rm fed})=P_{\rm \ac{fa}}(\delta_{\rm NP})$ which is
\begin{equation}\label{Q27}
    Q(\delta_{\rm fed},\tilde{\tau}_{\rm fed}) - Q(\delta_{\rm NP},\tilde{\tau}_{\rm NP}) =     \tilde{\tau}_{\rm fed}\Delta +(\tilde{\tau}_{\rm fed}-\tilde{\tau}_{\rm NP})P_{\rm MD}^\star\,. 
\end{equation}

We proceed now with two cases:

\paragraph*{Case 1: $\tilde{\tau}_{\rm fed} \leq \tilde{\tau}_{\rm NP}$} We see from \eqref{eq:obj_expanded} that we can bound
\begin{equation}
    Q(\delta_{\rm fed},\tilde{\tau}_{\rm fed}) - Q(\delta_{\rm NP},\tilde{\tau}_{\rm NP})\leq Q(\delta_{\rm fed},\tilde{\tau}_{\rm NP}) - Q(\delta_{\rm NP},\tilde{\tau}_{\rm NP})\,. 
\end{equation}
As they have the same integrand, the difference above depends only on the integration regions.
Let us define $\mathcal{S} = \mathcal{S}_1 \cup \mathcal{S}_2$ where $\mathcal{S}_1= (\delta_{\rm NP}\setminus \delta_{\rm fed}) $ and $\mathcal{S}_2=(\delta_{\rm fed} \setminus \delta_{\rm NP})$ as the region where the decision regions of $\delta_{\rm NP}$ and $\delta_{\rm fed}$ are different, as shown in Fig. \ref{fig:region_sets}. Then we can write
\begin{align}
&Q(\delta_{\rm fed},\tilde{\tau}_{\rm NP}) - Q(\delta_{\rm NP},\tilde{\tau}_{\rm NP})=\\
 &=\int_{\mathcal{S}_1} f_{\rm NP}(x) \, dx + \int_{\mathcal{S}_2} (-f_{\rm NP}(x)) \, dx 
= \int_{\mathcal{S}} |f_{\rm NP}(x)| \, dx\,,
\end{align}
as both terms are non-negative because $f_{\rm NP}(x) \ge 0$ inside $\delta_{\rm NP}$ (first integral) and $f_{\rm NP}(x) < 0$ outside it (second integral).

\begin{figure}
    \centering
    \resizebox{0.7\columnwidth}{!}{%
        \usetikzlibrary{patterns}
\begin{tikzpicture}[scale=1.2]

\def\NPblob{(0.5,1.5) .. controls (-1,2) and (-2,0) .. (-1.5,-1) .. controls (-1,-2) and (1,-1.5) .. (2,-1) .. controls (3,0) and (2,1) .. (0.5,1.5)}

\def\fedblob{(2,2) .. controls (0.5,2.5) and (0,1) .. (0.5,0) .. controls (1,-1) and (2.5,-2) .. (4,-1) .. controls (5.5,0) and (4,1.5) .. (2,2)}

\begin{scope}
    \clip \NPblob;
    \fill[pattern=north east lines, pattern color=blue!30] \fedblob [even odd rule] \NPblob;
\end{scope}

\begin{scope}
    \clip \fedblob;
    \fill[pattern=north west lines, pattern color=red!30] \NPblob [even odd rule] \fedblob;
\end{scope}

\draw[thick, blue] \NPblob;
\draw[thick, red] \fedblob;


\node[align=center, font=\footnotesize] at (1.3,0.3) {
    $\delta_{\rm NP} \cap \delta_{\rm fed}$ \\
    $f_{\rm NP} \ge 0$ \\
    $f_{\rm fed} \ge 0$
};

\node[blue!80!black, align=center, font=\footnotesize] at (-0.3,0) {
    $\mathcal{S}_1$ \\
    $f_{\rm NP} \ge 0$ \\
    $f_{\rm fed} < 0$
};

\node[red!80!black, align=center, font=\footnotesize] at (3.2,0) {
    $\mathcal{S}_2$ \\
    $f_{\rm fed} \ge 0$ \\
    $f_{\rm NP} < 0$
};

\node[font=\footnotesize, fill=white, inner sep=1pt] at (1.5,-2.2) {
    Outside: $f_{\rm NP} < 0$ and $f_{\rm fed} < 0$
};

\node[blue, font=\bfseries] at (-1, 1.8) {$\delta_{\rm NP}$};
\node[red, font=\bfseries] at (3.8, 1.8) {$\delta_{\rm fed}$};

\begin{scope}[shift={(-2,-3)}]
    \draw[pattern=north east lines, pattern color=black] (0,0) rectangle (0.4,0.3);
    \node[right, font=\scriptsize] at (0.4,0.15) {Disagreement Regions $\mathcal{S}_1$ and $\mathcal{S}_2$};
    
    \draw[thick, blue] (4,0.3) -- (4.4,0.3) node[right, black, font=\tiny] {$f_{\rm NP} = p_1^{1} - \tilde{\tau}_{\rm NP} p_1^{2}$};
    \draw[thick, red] (4,0) -- (4.4,0) node[right, black, font=\tiny] {$f_{\rm fed} = p^{1} - \tilde{\tau}_{\rm fed} p^2$};
\end{scope}

\end{tikzpicture}
    }
    \caption{Decision regions of federated and NP tests.}
    \label{fig:region_sets}
\end{figure}
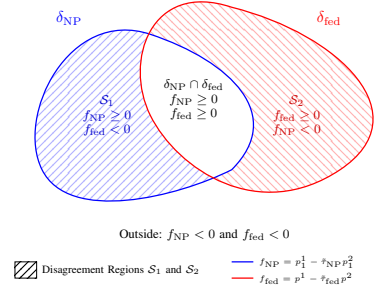

\subsubsection*{2. The Magnitude Inequality}

On the disagreement region $\mathcal{S}$, the functions $f_{\rm NP}(x)$ and $f_{\rm fed}(x)$ have opposite signs. For any two real numbers $a, b$ with opposite signs, it holds that $|a| \le |a - b|$. Applying this point-wise:
\begin{equation} \label{eq:diff_bound}
|f_{\rm NP}(x)| \le |f_{\rm NP}(x) - f_{\rm fed}(x)|, \quad \forall x \in \mathcal{S}.
\end{equation}
Thus, from \eqref{Q27} and \eqref{eq:diff_bound}, we have
\begin{equation}
\tilde{\tau}_{\rm fed}\Delta +(\tilde{\tau}_{\rm fed}-\tilde{\tau}_{\rm NP})P_{\rm MD}^\star \le \int_{\mathcal{S}} |f_{\rm NP}(x) - f_{\rm fed}(x)| \, dx.
\end{equation}
By applying the triangle inequality, we have 
\begin{equation}
\begin{split}
        &|f_{\rm NP}(x) - f_{\rm fed}(x)| \\ 
        &= |p_1^1 - p^1 + \tilde{\tau}_{\rm NP}  (p^2-p_1^2)| \\
        &\leq |p_1^1 - p^1| + \tilde{\tau}_{\rm NP} | (p^2-p_1^2)|.
\end{split} 
\end{equation}

\subsubsection*{3. Relation to Total Variation (TV) Distance}

Finally, we have, by definition of TV distance, 
\begin{align}
&\tilde{\tau}_{\rm fed}\Delta +(\tilde{\tau}_{\rm fed}-\tilde{\tau}_{\rm NP})P_{\rm MD}^\star \\ &\le \int_{\mathcal{S}} |p_1^1 - p^1| \, dx + \tilde{\tau}_{\rm NP} \int_{\mathcal{S}} |p_1^2 - p^2| \, dx \notag \\
&= 2 TV\big(p_1^1, p^1\big) + 2\tilde{\tau}_{\rm NP} TV\big(p_1^2, p^2\big).
\end{align}

Recalling that $p^1 = \frac{1}{2}(p^1_1 + p^1_2)$ and $p^2 = \frac{1}{2}(p^2_1 + p^2_2)$ are equal-weight mixtures, so by direct expansion of the $TV$ definition, $ TV(p^1_1, p^1) = \frac{1}{2} TV(p^1_1, p^1_2)$ and $TV(p^2_1, p^2) = \frac{1}{2} TV(p^2_1, p^2_2)$ we attain
\begin{equation}
\begin{split}
    P_{\rm MD}(\delta_{\rm fed}) \le \frac{1}{\tilde{\tau}_{\rm fed}} TV\big(p_1^1, p_2^{1} \big) +  \frac{\tilde{\tau}_{\rm NP}}{\tilde{\tau}_{\rm fed}} TV\big(p_1^2, p_2^{2} \big)+ \\
    (1+\frac{\tilde{\tau}_{\rm NP}-\tilde{\tau}_{\rm fed}}{\tilde{\tau}_{\rm fed}} )P_{\rm MD}^\star\,.
\end{split}
\end{equation}

\paragraph*{Case 2: $\tilde{\tau}_{\rm fed} \ge \tilde{\tau}_{\rm NP}$}
\begin{equation}
    Q(\delta_{\rm fed},\tilde{\tau}_{\rm fed}) - Q(\delta_{\rm NP},\tilde{\tau}_{\rm NP})\leq Q(\delta_{\rm fed},\tilde{\tau}_{\rm fed}) - Q(\delta_{\rm NP},\tilde{\tau}_{\rm fed})\,. 
\end{equation}
and with a similar reasoning, considering that now $(\tilde{\tau}_{\rm fed}-\tilde{\tau}_{\rm NP})P_{\rm MD}^\star \geq 0$, thus it can be discarded, we attain the bound 

\begin{equation}
\begin{split}
    P_{\rm MD}(\delta_{\rm fed}) & \leq
\frac{2}{\tilde{\tau}_{\rm fed}} TV\big(p_1^1, p^1\big) + 2 TV\big(p_1^2, p^2\big) + P_{\rm MD}^\star \\
& =
\frac{1}{\tilde{\tau}_{\rm fed}} TV\big(p_1^1, p_2^{1} \big) +  TV\big(p_1^2, p_2^{2}\big) + P_{\rm MD}^\star
\end{split}
\end{equation}
\subsubsection*{4. Final Performance Bound}
Combining both cases, and recalling that $\tau_k=\frac{1}{\tilde{\tau}_k}$, $k=\{{\rm fed, NP}\}$, we attain \eqref{eq:final_bound}. 

\clearpage

\end{document}